\documentclass[11pt,reqno]{amsart}
\usepackage{amsmath}
\usepackage{amssymb}
\usepackage{verbatim}
\usepackage{subfigure}
\usepackage{booktabs}
\usepackage{epsfig}
\usepackage{pst-plot,pstricks-add}

\usepackage{wavelet}

\newcommand{\das}{\operatorname{DAS}}
\newcommand{\dtcwt}{\operatorname{DT-\mathbb{C}WT}}
\newcommand{\tpctf}{\operatorname{TP-\mathbb{C}TF}}
\newcommand{\ctf}{\operatorname{\mathbb{C}TF}}
\newcommand{\uft}{\operatorname{UFT}}
\newcommand{\sampled}{\,!\,} 
\newcommand{\reduced}{\downarrow}

\usepackage[top=0.7in,bottom=0.7in,left=0.7in,right=0.7in]{geometry}
\date{}

\newtheorem{lemma}{Lemma}

\newtheorem{theorem}[lemma]{Theorem}

\numberwithin{equation}{section}

\begin{document}

\title[Directional Tensor Product Complex Tight Framelets]{Directional Tensor Product Complex Tight Framelets with\\Low Redundancy}

\author{Bin Han}
\address{Department of Mathematical and Statistical Sciences,
University of Alberta, Edmonton, Alberta T6G 2G1, Canada.\newline
\quad  {\tt bhan@ualberta.ca}\qquad
 {\tt http://www.ualberta.ca/$\sim$bhan}
}

\author{Zhenpeng Zhao}
\address{Department of Mathematical and Statistical Sciences,
University of Alberta, Edmonton, Alberta T6G 2G1, Canada. \newline
{\tt zzhao7@ualberta.ca}
}

\author{Xiaosheng Zhuang}

\address{Department of Mathematics, City University of Hong Kong, Tat Chee Avenue, Kowloon Tong, Hong Kong.\newline
 {\tt xzhuang7@cityu.edu.hk}
}

\thanks{Research of Bin Han and Zhenpeng Zhao was supported in part by
the Natural Sciences and Engineering Research Council of Canada (NSERC Canada under
Grant 05865). Research of Xiaosheng Zhuang was supported by Research Grants Council of Hong Kong (Project No. CityU 108913).}

\thanks{Contact information of corresponding author Bin Han: E-mail: bhan@ualberta.ca, Phone: 1-780-4924289, Fax: 1-780-4926826,  Web: http://www.ualberta.ca/$\sim$bhan}

\makeatletter \@addtoreset{equation}{section} \makeatother
\begin{abstract}
Having the advantages of redundancy and flexibility, various types of tight frames have already shown impressive performance in applications such as image and video processing. For example, the undecimated wavelet transform, which is a particular case of tight frames, is known to have good performance for the denoising problem. Empirically, it is widely known that higher redundancy rate of a (tight) frame often leads to better performance in applications. The wavelet/framelet transform is often implemented in an undecimated fashion for the purpose of better performance in practice.
Though high redundancy rate of a tight frame can improve performance in applications,
as the dimension increases, it also makes the computational cost
skyrocket and the storage of frame coefficients increase exponentially.
This seriously restricts the usefulness of such tight frames for problems in moderately high dimensions such as video processing in dimension three.
Inspired by the directional tensor product complex tight framelets $\tpctf_m$ with $m\ge 3$ in \cite{Han:mmnp:2013,HanZhao:siims:2014} and their impressive performance for image processing in \cite{HanZhao:siims:2014,SHB:2014}, in this paper we introduce a directional tensor product complex tight framelet $\tpctf^{\reduced}_6$ (called reduced $\tpctf_6$) with low redundancy. Such $\tpctf_6^{\reduced}$ is a particular example of tight framelet filter banks with mixed sampling factors.
The $\tpctf^{\reduced}_6$ in $d$ dimensions
not only offers good directionality but also has the low redundancy rate $\frac{3^d-1}{2^d-1}$ (e.g., the redundancy rates are $2, 2\mathord{\frac{2}{3}}, 3\mathord{\frac{5}{7}}, 5\mathord{\frac{1}{3}}$ and $7\mathord{\frac{25}{31}}$ for dimension $d=1,\ldots, 5$, respectively).
Moreover, our numerical experiments on image/video denoising and inpainting show that
the performance using our proposed $\tpctf^{\reduced}_6$
is often comparable or sometimes better than several state-of-the-art frame-based methods which have much higher redundancy rates than that of $\tpctf^{\reduced}_6$.
\end{abstract}

\keywords{Directional complex tight framelets, tight framelet filter banks with mixed sampling factors, redundancy rate, tensor product, image and video denoising, image and video inpainting}

\subjclass[2010]{42C40, 42C15, 94A08, 68U10} \maketitle

\pagenumbering{arabic}

\section{Introduction and Motivations}

Though wavelets have many useful applications, they have several shortcomings in dealing with multidimensional problems. For example, tensor product real-valued wavelets are known for lack of the desired properties of translation invariance and directionality (\cite{CoifDo,KL,SBK}).
There are a lot of papers in the current literature to improve the performance of classical tensor product (i.e., separable) real-valued wavelets by remedying these two shortcomings. In one direction, translation invariance of wavelets can be improved by using wavelet frames instead of orthonormal wavelets (see \cite{CHS,CoifDo,DHRS,Han:acha:1997,Han:acha:2012,Han:acha:2013,Han:acha:2014,RonShen:twf,SBK} and many references therein).
For example, the undecimated wavelet transform (\cite{CoifDo}) using Daubechies orthonormal wavelets has been known to be effective for the denoising problem. In fact, such an undecimated wavelet transform employs a particular case of tight frames with high redundancy.
Comparing with an orthonormal basis, a (tight) frame is more general and has redundancy by allowing more elements into its system. The added redundancy of a tight frame not only improves the property of translation invariance but also makes the design of a tight frame more flexible (see \cite{CHS,CoifDo,DHRS,Han:acha:1997,Han:acha:2012,Han:acha:2013,Han:acha:2014,HanZhuang:acha:2014,RonShen:twf} and references therein).
In the other direction, many papers in the literature have been studying directional representations, to only mention a few here, curvelets in \cite{FDCT}, contourlets in \cite{DV}, shearlets in \cite{HanZhuang:acha:2014,KL,KSZ,LimIEEE, LimIEEE2} and many references therein, surfacelets in \cite{LD}, dual tree complex wavelet transform in \cite{K:acha:2001,SBK,SS:bslocal}, complex tight framelets in \cite{Han:acha:2012,Han:mmnp:2013,Han:acha:2014,HanZhao:siims:2014}, plus many other directional representations. To improve directionality of tensor product real-valued wavelets, due to the requirement of the additional angular resolution for a directional representation, it is almost unavoidable to employ a (tight) frame instead of an orthonormal basis by allowing redundancy. In fact, to our best knowledge, all representations having either directionality and/or (near) translation invariance, which have been known in the literature so far, use either a frame or a tight frame with various degrees of redundancy.

In the following, let us explain by what we mean the redundancy rate of a transform or a system.
To this end, let us recall the definition of a tight framelet filter bank.
For $u=\{u(k)\}_{k\in \dZ} \in \dlp{1}$, we define the \emph{Fourier series} (or \emph{symbol}) $\wh{u}$ of the sequence $u$ to be $\wh{u}(\xi):=\sum_{k\in \dZ} u(k) e^{-ik\cdot \xi}$, $\xi\in \dR$. Note that $\wh{u}$ is a $2\pi\dZ$-periodic function satisfying $\wh{u}(\xi+2\pi k)=\wh{u}(\xi)$ for all $k\in \dZ$.
For $a,b_1, \ldots, b_s\in \dlp{1}$, we say that $\{a; b_1, \ldots, b_s\}$ is a ($d$-dimensional dyadic) \emph{tight framelet filter bank} if
\[
|\wh{a}(\xi)|^2+\sum_{\ell=1}^s |\wh{b_\ell}(\xi)|^2=1 \quad \mbox{and}\quad
\wh{a}(\xi)\ol{\wh{a}(\xi+\pi \omega)}+\sum_{\ell=1}^s \wh{b_\ell}(\xi) \ol{\wh{b_\ell}(\xi+\pi \omega)}=0, \qquad \forall\; \omega\in ([0,1]^d\cap \dZ)\backslash\{0\}
\]
for almost every $\xi\in \dR$. Moreover, a ($d$-dimensional dyadic) tight framelet filter bank $\{a; b_1, \ldots, b_s\}$ with $s=2^d-1$ is called a ($d$-dimensional dyadic) \emph{orthonormal wavelet filter bank}. A $d$-dimensional tight framelet (or orthonormal wavelet) filter bank is often obtained through tensor product. For one-dimensional filters $u_1, \ldots, u_d\in \lp{1}$, we define their $d$-dimensional tensor product filter $u_1\otimes \cdots \otimes u_d$ by $(u_1\otimes \cdots \otimes u_d)(k_1,\ldots, k_d):=u_1(k_1)\cdots u_d(k_d)$ for $k_1, \ldots, k_d\in \Z$. In particular, we define $\otimes ^d u:=u\otimes \cdots \otimes u$ with $d$ copies of $u$.
If $\{a; b_1,\ldots,b_s\}$ is a one-dimensional (dyadic) tight framelet filter bank (or an orthonormal wavelet filter bank with $s=1$), then it is straightforward to check that $\otimes^d \{a; b_1, \ldots, b_s\}$ is a $d$-dimensional dyadic tight framelet (or orthonormal wavelet) filter bank.
See \cite{CHS,DHRS,Han:acha:1997,Han:acha:2010,Han:acha:2012,Han:acha:2013,Han:acha:2014,RonShen:twf}
and Section~\ref{sec:tffb:ms} for connections of tight framelet filter banks with tight framelets in $L_2(\R^d)$.

A fast wavelet/framelet transform is implemented through the operations of convolution and sampling. Let $v\in \dlp{\infty}$ be a $d$-dimensional input signal and
let $u$ be a filter from a given $d$-dimensional tight framelet filter bank $\{a; b_1, \ldots, b_s\}$.
Roughly speaking, for the decomposition/forward transform,
the data $v$ is first convolved with the flip-conjugate filter $u^\star$ (that is, $u^\star(k):=\ol{u(-k)}, k\in\dZ$) as $v*u^\star:=\sum_{k\in \dZ} v(k) u^\star(\cdot-k)$ and then it is downsampled as $w:=(v*u^\star) \ds 2I_d:= (v*u^\star)(2\cdot)$, where $w$ is called the sequence of frame coefficients.
The decomposition transform can be applied recursively $J$ times with $v$ being replaced by $(v*a^\star)\ds 2I_d$ (that is, $u=a$) as the new input data, where $J\in \N$ is the decomposition level. For the reconstruction/backward transform, the coefficient sequence $w$ is upsampled as $(w\us 2I_d)(k):=w(k/2)$ if $k\in 2\dZ$ and $(w\us 2I_d)(k):=0$ for $k\in \dZ\bs [2\dZ]$,
and then it is convolved with $u$ as $(w\us 2I_d)*u$. Finally, all the reconstructed sequences are added together as one reconstructed data. See Figure~\ref{fig:ffrt} for an illustration of a two-level fast framelet transform employing a one-dimensional dyadic tight framelet filter bank $\{a; b_1, \ldots, b_s\}$ (but with $\sqrt{4},\, \ds 4,\, \us 4$ in Figure~\ref{fig:ffrt} being replaced by $\sqrt{2}, \,\ds 2,\, \us 2$, respectively). See Section~\ref{sec:tffb:ms} for more details on a fast framelet transform.

Most $d$-dimensional problems and data in applications have finite length. For a given real-valued data $v$ of finite length, one first extends it into a periodic sequence $v^e$ on $\dZ$. Then one performs a wavelet/framelet transform on the extended data $v^e$. This induces a linear transform on the original data $v$ and the decomposition transform can be rewritten using a matrix $\wa$.
More precisely, we can arrange the $d$-dimensional real-valued data $v$ properly so that it can be regarded as an $n\times 1$ column vector in $\R^n$, that is, $v\in \R^n$. Performing a linear transform $\wa$ on $v$, we obtain another column vector $w:=\wa v\in \R^N$ of frame coefficients. If $\{a; b_1, \ldots, b_s\}$ with $s=2^d-1$ is a real-valued orthonormal wavelet filter bank, then $N=n$ and $\wa$ is a real-valued $n\times n$ orthogonal matrix satisfying $\wa^\tp \wa=I_n$. If $\{a;b_1, \ldots, b_s\}$ is a real-valued tight framelet filter bank, then we must have $N\ge n$ and $\wa$ is a real-valued $N\times n$ matrix satisfying $\wa^\tp \wa=I_n$. Therefore, we call the ratio $N/n$ the \emph{redundancy rate} of the linear transform $\wa$ or its underlying tight frame, since it is the ratio between the $N$ number of frame coefficients over the $n$ number of original input data. Also note that the redundancy rate $N/n$ is independent of the length $n$ of input data and depends only on the number $s$ of high-pass filters and the sampling factor (which is $2I_d$ here).

We now look at the redundancy rate of an undecimated wavelet/framelet transform (denoted by $\uft_s$) using tensor products of a one-dimensional real-valued tight framelet filter bank $\{a; b_1, \ldots, b_s\}$ (when $s=1$, it is an orthonormal wavelet filter bank and $\uft_1$ becomes $\operatorname{UWT}$--the undecimated wavelet transform).
Here the word undecimated means that the upsampling and downsampling operations in a standard wavelet/framelet transform are completely removed.
Undecimated framelet transforms using spline tight framelet filter banks $\{a^B_2; \mathring{b}_1, \mathring{b}_2\}$ and $\{a^B_4; b_1,b_2,b_3,b_4\}$ with  $\wh{a^B_m}(\xi):=2^{-m}(1+e^{-i\xi})^m$
have applications to many image restoration problems as reported in \cite{CCS:framelet:inpainting,CCS:ipi:2010,CDOS,CRSS,DongShen:tutorial,LiShenSuter:ieee:2013}
and many references therein.
The tensor product $d$-dimensional tight framelet filter bank is $\otimes^d \{a; b_1, \ldots, b_s\}$ which consists of one real-valued low-pass filter $\otimes^d a$ and $(s+1)^d-1$ real-valued high-pass filters.
If the decomposition level is $J\in \N$, since all the filters are implemented in an undecimated fashion, the redundancy rate of the $d$-dimensional undecimated framelet transform using the tensor product real-valued tight framelet filter bank $\otimes^d\{a; b_1, \ldots, b_s\}$
is $((s+1)^d-1)J+1$. To take advantages of the multiscale structure of wavelets, it is necessary that the decomposition level $J$ should be as high as possible. For example, for a standard $512\times 512$ grayscale image, the wavelet decomposition level is often set to be at least $J=5$ (note that $512=2^9$). Let us here take a moderate choice of $J=3$ and use the smallest $s=1$ (that is, we are using an orthonormal wavelet filter bank). For dimension $d=3$ and $J=3$, the redundancy rate of a tensor product undecimated wavelet transform is $22$. However, as we mentioned before, tensor product real-valued orthonormal wavelets lack directionality and translation invariance. To improve directionality or translation invariance, we must use a tight framelet filter bank with $s\ge 2$.
For $d=3$ and $J=3$, the redundancy rates of $\uft$s
are $22,79,190,373,646$ for $s=1,\ldots, 5$, respectively. 
See Table~\ref{tab:redundancy} for a numerical illustration on redundancy rates of an undecimated wavelet/framelet transform.

By employing a pair of two correlated one-dimensional real-valued orthonormal wavelet filter banks, the dual tree complex wavelet transform ($\dtcwt$) offers directionality and translation invariance with the redundancy rate $2^d$ in $d$ dimensions for any decomposition level $J\in \N$.
See \cite{K:acha:2001,SBK,SS:bslocal} and \cite[Section~2]{HanZhao:siims:2014} as well as references therein for more details on $\dtcwt$. One-dimensional finitely supported complex-valued tight framelet filter banks have been extensively studied in \cite{Han:acha:2013,Han:acha:2014}.
A family of directional tensor product complex tight framelet filter banks ($\tpctf$) has been initially introduced in \cite{Han:mmnp:2013} and further developed in \cite{HanZhao:siims:2014} for the purpose of image denoising. The family of one-dimensional complex tight framelet filter banks introduced and used in \cite{Han:mmnp:2013,HanZhao:siims:2014} is $\ctf_m$, where $m\ge 3$ is the total number of filters in $\ctf_m$.
The low-pass filter in $\ctf_m$ is real-valued but its high-pass filters are complex-valued. If $m$ is odd, then the $d$-dimensional tensor product tight framelet filter bank $\tpctf_m$ has one real-valued low-pass filter and $(m-1)^d-1$ complex-valued  high-pass filters.
Consequently, its redundancy rate is no more than $\frac{m^d-1}{2^d-1}$ for dimension $d$ and for any decomposition level $J\in \N$.
If $m$ is even, then the $d$-dimensional tensor product tight framelet filter bank $\tpctf_m$ has one real-valued low-pass filter and $m^d-2^d$ complex-valued high-pass filters. Therefore, its redundancy rate is no more than $\frac{m^d-2^d}{2^d-1}$ for dimension $d$ and for any decomposition level $J\in \N$. For both the dual tree complex wavelet transform $\dtcwt$ and the tensor product complex tight framelets $\tpctf_m$, a complex frame coefficient is counted as two real frame coefficients in the calculation of their redundancy rates. See Section~\ref{sec:tpctf} for more detailed explanation about the redundancy rates of $\tpctf_m$.
The frequently used tensor product complex tight framelets for image denoising in \cite{HanZhao:siims:2014} are $\tpctf_4$ and $\tpctf_6$.
The $\tpctf_4$ has almost the same performance, directionality and redundancy rate
as those of $\dtcwt$. The $\tpctf_6$ has much better performance than $\tpctf_4$ and $\dtcwt$ for image denoising in \cite{HanZhao:siims:2014} and image inpainting in \cite{SHB:2014}, but it has higher redundancy rate $\frac{6^d-2^d}{2^d-1}$ for dimension $d$. See Table~\ref{tab:redundancy} for some numerical illustration on redundancy rates of $\tpctf_m$.
See \cite{Han:mmnp:2013,HanZhao:siims:2014} as well as Section~\ref{sec:tpctf} for more detailed discussion on directional tensor product complex tight framelets and their redundancy rates.

Beyond the above tensor product (i.e., separable) transforms for multidimensional problems, to achieve directionality, there are also many nonseparable approaches. We shall use the notation $d$D to stand for $d$ dimensions or $d$-dimensional. Some examples of such nonseparable transforms are 2D curvelets in \cite{FDCT}, 2D contourlets in \cite{DV}, 2D and 3D shearlets in (\cite{HanZhuang:acha:2014,KL,KSZ,LimIEEE,LimIEEE2} and references therein), 3D surfacelets in \cite{LD}, and directional tight framelets in \cite{Han:acha:1997,Han:acha:2012,HanZhuang:acha:2014}, as well as quite a few more nonseparable transforms in the literature.
The redundancy rates of such nonseparable transforms depend on the choices of the numbers of directions at each resolution level and the decomposition level $J\in \N$.
Generally speaking, to achieve reasonable performance in applications,
those nonseparable transforms often have much higher redundancy rates than those of the tensor product transforms using the dual tree complex wavelet transform and directional complex tight framelets. See Section~\ref{sec:experiments} for the redundancy rates and performance of several nonseparable transforms using directional representations.

Though empirically higher redundancy rate of a tight frame often leads to better performance in applications, the computational costs increase exponentially with respect to higher redundancy rate and dimensionality. This causes serious constraints on computational expenses and storage requirement for multidimensional problems. To our best knowledge, most of the above mentioned directional representations and tight frames can achieve reasonably good performance with computational costs being manageable by a standard PC for two-dimensional problems. However, for applications in three or higher dimensions such as video processing, the expensive computational cost
becomes a serious issue, without even mentioning the fact that one often tends to increase the redundancy rates in order
to achieve reasonably good performance for applications in three or higher dimensions.
This difficulty seriously restricts the usefulness of such tight frames and directional representations for multidimensional problems (in particular, for problems in moderately high dimensions such as video processing in dimension three).
Motivated by the approach of directional tensor product complex tight framelets in \cite{Han:mmnp:2013,HanZhao:siims:2014}, to remedy the above mentioned difficulty,
in this paper we shall construct a tight wavelet frame having the following desired properties:

\begin{enumerate}
\item[(i)] The tight frame is obtained through the tensor product of a one-dimensional tight framelet filter bank.
\item[(ii)] The tight frame has low redundancy rate and all its high-pass elements have good directionality.
\item[(iii)] The tight frame has good performance for applications such as denoising and inpainting, comparing with more complicated directional representations and tight frames with much higher redundancy rates.
\end{enumerate}

The tensor product structure in item (i) and low redundancy rate in item (ii) of such a tight frame make it computationally efficient and attractive, while low redundancy also significantly reduces the storage requirement for frame coefficients. Good directionality in item (ii) is needed in order to have good performance as required in item (iii).
In this paper we shall achieve all the above goals by modifying the construction of directional tensor product complex tight framelet filter banks $\tpctf_m$ with $m\ge 3$ in \cite{Han:mmnp:2013,HanZhao:siims:2014}. Though our approach can be easily applied to all $\tpctf_m$, for simplicity of presentation, in this paper we restrict our attention to one particular example: the directional tensor product complex tight framelet $\tpctf_6$, whose underlying one-dimensional tight framelet filter bank is $\ctf_6$.
As demonstrated in \cite{HanZhao:siims:2014} for image denoising and in \cite{SHB:2014} for image inpainting, this $\tpctf_6$ has much better performance than $\dtcwt$, $\tpctf_4$, curvelets, 2D shearlets, real-valued spline tight frames, discrete cosine transform, and many other frame-based methods. We are hoping to be able to significantly reduce the redundancy rate of $\tpctf_6$ while trying to keep almost all the desirable properties of $\tpctf_6$.
As a consequence, we denote our modified directional tensor product complex tight framelet by $\tpctf^{\reduced}_6$ and call it (redundancy) reduced $\tpctf_6$, where the superscript $\reduced$ here means that $\tpctf^{\reduced}_6$ is a reduced (or further downsampled) version of $\tpctf_6$ by decreasing its redundancy rate while trying to keep all the good properties of the original $\tpctf_6$. As we shall see in Section~\ref{sec:tpctf}, the redundancy rate of $\tpctf_6^{\reduced}$ is $\frac{3^d-1}{2^d-1}$ for dimension $d$ and for any decomposition level $J\in \N$, while as we discussed before, the redundancy rate of $\tpctf_6$ is $\frac{6^d-2^d}{2^d-1}=2^d \times \frac{3^d-1}{2^d-1}$ (that is, the redundancy rate of $\tpctf_6$ is $2^d$ times that of $\tpctf_6^{\reduced}$ in dimension $d$).
See Table~\ref{tab:redundancy} for an illustration and comparison of redundancy rates of various tight frames.

\begin{table}[ht]
\begin{center}
\begin{tabular}{|c||c|c|c|c|c|c|c|c|c|c||} \hline
$d$        &$1$ &$2$ &$3$ &$4$ &$5$ &$6$ &$7$ &$8$ &$9$ &$10$ \\ \hline
$\operatorname{UWT}$ &$4$ &$10$ &$22$ &$46$ &$94$ &$190$ &$383$ &$766$
&$1534$ &$3070$
\\ \hline
$\uft_2$ &$7$ &$25$ &$79$ &$241$ &$727$ &$2185$ &$6559$
&$19681$ &$59047$ &$177145$
\\ \hline
$\uft_4$ &$13$ &$73$ &$373$ &$1873$ &$9373$ &$46873$ &$234373$
&$1171873$ &$5859373$ &$29296873$
\\ \hline
$\dtcwt$ &$2$  &$4$  &$8$ &$16$ &$32$ &$64$ &$128$ &$256$ &$512$ &$1024$
\\ \hline
$\tpctf_3$ &$2$ &$2\mathord{\frac{2}{3}}$ &$3\mathord{\frac{5}{7}}$
&$5\mathord{\frac{1}{3}}$  &$7\mathord{\frac{25}{31}}$
&$11\mathord{\frac{5}{9}}$ &$17\mathord{\frac{27}{127}}$
&$25\mathord{\frac{37}{51}}$  &$38\mathord{\frac{264}{511}}$
&$57\mathord{\frac{67}{93}}$\\ \hline
$\tpctf_4$ &$2$  &$4$  &$8$ &$16$ &$32$ &$64$ &$128$ &$256$ &$512$ &$1024$
\\ \hline
$\tpctf_5$ &$4$ &$8$ &$17\mathord{\frac{5}{7}}$  &$41\mathord{\frac{3}{5}}$
&$100\mathord{\frac{24}{31}}$ &$248$ &$615\mathord{\frac{19}{127}}$
&$1531\mathord{\frac{73}{85}}$ &$3822\mathord{\frac{82}{511}}$
&$9546\mathord{\frac{2}{31}}$
\\ \hline
$\tpctf_6$ &$4$      &$10\mathord{\frac{2}{3}}$   &$29\mathord{\frac{5}{7}}$
&$85\mathord{\frac{1}{3}}$  &$249\mathord{\frac{25}{31}}$  &$739\mathord{\frac{5}{9}}$ &$2203\mathord{\frac{27}{127}}$ &$6585\mathord{\frac{37}{51}}$ &$19720\mathord{\frac{264}{511}}$
&$59105\mathord{\frac{67}{93}}$
\\ \hline
$\tpctf_6^{\reduced}$ &$2$ &$2\mathord{\frac{2}{3}}$ &$3\mathord{\frac{5}{7}}$
&$5\mathord{\frac{1}{3}}$  &$7\mathord{\frac{25}{31}}$
&$11\mathord{\frac{5}{9}}$ &$17\mathord{\frac{27}{127}}$
&$25\mathord{\frac{37}{51}}$  &$38\mathord{\frac{264}{511}}$
&$57\mathord{\frac{67}{93}}$\\ \hline
\end{tabular}
\medskip
\begin{caption}{Comparison of redundancy rates of various tight frames for different dimensions $d$.
$\operatorname{UWT}$ is the undecimated wavelet transform with decomposition level $J=3$ and using the tensor product of a 1D real-valued orthonormal wavelet filter bank $\{a;b\}$.
$\uft_s$ is the undecimated framelet transform with decomposition level $J=3$ and using the tensor product of a 1D real-valued tight framelet filter bank $\{a;b_1,\ldots,b_s\}$ (Hence, $\operatorname{UWT}$ is just $\uft_1$).
$\dtcwt$ is the dual tree complex wavelet transform. $\tpctf_m$ is the tensor product complex tight framelet with $m=3,4,5,6$. $\tpctf_6^{\reduced}$ is our proposed tensor product complex tight framelet with low redundancy. It is interesting to point out here that $\tpctf_6^{\reduced}$ has the same low redundancy rate as $\tpctf_3$, but $\tpctf_6^{\reduced}$ enjoys the same directionality as $\tpctf_6$.}
\end{caption}\label{tab:redundancy}
\end{center}
\end{table}

The structure of the paper is as follows. In order to study tensor product complex tight framelets with low redundancy,
in Section~\ref{sec:tffb:ms} we shall generalize the notion of dyadic tight framelet filter banks by introducing tight framelet filter banks with mixed sampling factors. Then we shall study their various properties and fast framelet transforms of such tight framelet filter banks with mixed sampling factors in Section~\ref{sec:tffb:ms}.
In Section~\ref{sec:tpctf}, we shall recall the tensor product complex tight framelet filter banks $\tpctf_m$ and their underlying one-dimensional complex tight framelet filter banks $\ctf_m$ with $m\ge 3$ from \cite{Han:mmnp:2013,HanZhao:siims:2014}.
Then we shall discuss the redundancy rates of $\tpctf_m$.
Next we shall provide details on our construction of directional tensor product complex tight framelet $\tpctf^{\reduced}_6$ with low redundancy. Such $\tpctf_6^{\reduced}$ is a particular example of tight framelet filter banks with mixed sampling factors in Section~\ref{sec:tffb:ms}.
Though our approach can be easily applied to all $\tpctf_m$ with $m\ge 3$, for simplicity of presentation, we only deal with $\tpctf_6$ in Section~\ref{sec:tpctf}.
In Section~\ref{sec:experiments}, we shall test the performance of our proposed directional complex tight framelet $\tpctf^{\reduced}_6$ with low redundancy rate and compare its performance with several state-of-the-art frame-based methods. Our numerical experiments on image/video denoising and inpainting show that the performance using our tensor product directional complex tight framelet $\tpctf^{\reduced}_6$ with low redundancy is often comparable or sometimes better than several state-of-the-art frame-based methods which often have much higher redundancy rates.
Moreover, our numerical experiments show that $\tpctf_6^{\reduced}$ is particularly effective for images and videos having rich textures.

\section{Tight framelet filter banks with mixed sampling factors}
\label{sec:tffb:ms}

In this section we shall introduce tight framelet filter banks with mixed sampling factors and then study their properties.
As we shall see later in Section~\ref{sec:tpctf}, our proposed directional tensor product complex tight framelet $\tpctf_6^{\reduced}$ with low redundancy is a particular case of tight framelet filter banks with mixed sampling factors.

\subsection{Fast framelet transform using tight framelet filter banks with mixed sampling factors}

Our key idea to derive a directional tight framelet with low redundancy from the tensor product complex tight framelet filter banks $\tpctf_m$ in \cite{Han:mmnp:2013,HanZhao:siims:2014} is to use higher sampling factors such as $4I_d$ instead of $2 I_d$. To this end, let us generalize the definition of a ($d$-dimensional dyadic) tight framelet filter bank $\{a; b_1, \ldots, b_s\}$, which uses the uniform sampling matrix $2I_d$, where $I_d$ is the $d\times d$ identity matrix.

Let $\dm$ be a $d\times d$ invertible integer matrix. For a sequence $u=\{u(k)\}_{k\in \dZ}: \dZ \rightarrow \C$, the downsampling sequence $u\ds \dm$ and the upsampling sequence $u\us \dm$ with the sampling matrix $\dm$ are defined by
\[
[u\ds \dm](k):=u(\dm k), \quad k\in \dZ
\quad \mbox{and}\quad
[u\us \dm](k):=\begin{cases} u(\dm^{-1} k), &\text{if $k\in \dm\dZ$},\\
0, &\text{if $k\in \dZ \bs [\dm \dZ]$}.\end{cases}
\]
We call $\dm$ the \emph{sampling factor or matrix}.
To explicitly specify the sampling matrix $\dm$ associated with a filter $u$, we shall adopt the notation $u\sampled \dm$. Under the new notation, a ($d$-dimensional dyadic) tight framelet filter bank $\{a; b_1,\ldots, b_s\}$ will be denoted more precisely as
$\{a\sampled 2I_d; b_1 \sampled 2I_d,\ldots, b_s\sampled 2I_d\}$, since the sampling matrix is uniformly $2I_d$.

For $1\le p<\infty$, $\dlp{p}$ consists of all the sequences $v: \dZ \rightarrow \C$ satisfying $\|v\|_{\dlp{p}}^p:=\sum_{k\in \dZ} |v(k)|^p<\infty$. Similarly, $v\in \dlp{\infty}$ if $\|v\|_{\dlp{\infty}}:=\sup_{k\in \dZ} |v(k)|<\infty$.

A discrete framelet transform can be described using the subdivision operator and the transition operator. For a filter $\tu \in \dlp{1}$ and a $d\times d$ integer matrix $\dm$, the \emph{subdivision operator} $\sd_{\tu, \dm}: \dlp{\infty}\rightarrow \dlp{\infty}$ and the \emph{transition operator} $\tz_{\tu,\dm}: \dlp{\infty}\rightarrow \dlp{\infty}$ are defined to be
\begin{align*}
[\sd_{u,\dm} v](n)&:=|\det(\dm)| \sum_{k\in \dZ} v(k) u(n-\dm k), \qquad n\in \dZ,\\
[\tz_{u,\dm} v](n)&:=|\det(\dm)|\sum_{k\in \dZ} v(k) \ol{\tu(k-\dm n)}, \qquad n\in \dZ,
\end{align*}
for $\tv\in \dlp{\infty}$. Define $\Omega_{\dm}:=[\dm^{-\tp}\dZ]\cap [0,1)^d$.
In terms of Fourier series, we have
\begin{equation}\label{def:sd:tz:F}
\wh{\sd_{\tu,\dm} \tv}(\xi)=|\det(\dm)| \hv(\dm^\tp \xi) \hu(\xi),\qquad
\wh{\tz_{\tu, \dm}\tv}(\xi)=\sum_{\omega\in \dmfc} \hv(\dm^{-\tp}\xi+2\pi \omega)
\ol{\hu(\dm^{-\tp}\xi+2\pi \omega)}.
\end{equation}
Define the flip-conjugate sequence $u^\star$ of $u$ by $u^\star(k):=\ol{u(-k)}, k\in \dZ$, that is, $\wh{u^\star}(\xi)=\ol{\wh{u}(\xi)}$. Then $\sd_{u,\dm} v=|\det (\dm)| (v\us \dm)*u$ and
$\tz_{u,\dm} v=|\det(\dm)| (v*u^\star)\ds \dm$, where $v*u:=\sum_{k\in \dZ} v(k)u(\cdot-k)$ is the convolution of $v$ and $u$.

Let $a, b_1, \ldots, b_s\in \dlp{1}$ and let $\dm, \dm_1, \ldots, \dm_s$ be $d\times d$ invertible integer matrices.
For $J\in \N$, we now describe a $J$-level ($d$-dimensional) discrete/fast framelet transform
employing a filter bank $\{a\sampled \dm; b_1 \sampled \dm_1, \ldots, b_s\sampled \dm_s\}$.
For a given data $v_0\in \dlp{\infty}$, the $J$-level discrete framelet decomposition (or forward transform) employing the filter bank $\{a\sampled \dm; b_1 \sampled \dm_1, \ldots, b_s\sampled \dm_s\}$ is
\begin{equation}\label{def:wa}
v_j:=|\det(\dm)|^{-1/2} \tz_{a,\dm} v_{j-1}
\quad \mbox{and}\quad w_{\ell,j}:=|\det(\dm_\ell)|^{-1/2} \tz_{ b_\ell, \dm_\ell} v_{j-1}, \qquad \ell=1, \ldots, s,\, j=1,\ldots,J,
\end{equation}
where $v_j$ are called sequences of low-pass coefficients and
all $w_{\ell,j}$ are called sequences of high-pass coefficients of the input signal $v_0$.
A $J$-level discrete framelet reconstruction (or backward transform) employing the filter bank
$\{a\sampled \dm; b_1 \sampled \dm_1, \ldots, b_s\sampled \dm_s\}$
can be described by
\be \label{def:ws}
\mathring{v}_{j-1}:=|\det(\dm)|^{-1/2}\sd_{a,\dm} \mathring{v}_j+ \sum_{\ell=1}^s |\det(\dm_\ell)|^{-1/2} \sd_{b_\ell,\dm_\ell} \mathring{w}_{\ell,j},\qquad j=J,\ldots,1,
\ee
where $\mathring{v}_0$ is a reconstructed sequence on $\dZ$. The \emph{perfect reconstruction property} requires that the reconstructed sequence $\mathring{v}_0$ should be exactly the same as the original input data $v_0$ if $\mathring{v}_J=v_J$ and $\mathring{w}_{\ell,j}=w_{\ell,j}$ for $j=1,\ldots, J$ and $\ell=1,\ldots, s$.
See Figure~\ref{fig:ffrt} for an illustration of a two-level fast framelet transform using a one-dimensional tight framelet filter bank $\{a \sampled 2; b_1 \sampled 4,\ldots, b_s\sampled 4\}$.

Using \cite[Theorem~2.1]{Han:mmnp:2013},
we have the following result on the perfect reconstruction property of a filter bank $\{a\sampled \dm; b_1 \sampled \dm_1, \ldots, b_s\sampled \dm_s\}$.

\begin{theorem}\label{thm:dft:pr}
Let $a, b_1, \ldots, b_s\in \dlp{1}$ and let $\dm, \dm_1, \ldots, \dm_s$ be $d\times d$ invertible integer matrices.
Then the following statements are equivalent to each other:
\begin{enumerate}
\item[(i)] For every $J\in \N$, the $J$-level fast framelet transform employing the filter bank $\{a\sampled \dm; b_1 \sampled \dm_1, \ldots, b_s\sampled \dm_s\}$ has perfect reconstruction property.
\item[(ii)] The one-level discrete framelet transform employing the filter bank $\{a\sampled \dm; b_1 \sampled \dm_1, \ldots, b_s\sampled \dm_s\}$ has \emph{perfect reconstruction property}, that is, for all $\tv\in \dlp{\infty}$,
\begin{equation}\label{pr}
v=
|\det(\dm)|^{-1}\sd_{a,\dm}\tz_{a,\dm} v+ |\det(\dm_1)|^{-1} \sd_{b_1,\dm_1}\tz_{b_1,\dm_1} v+\cdots+|\det(\dm_s)|^{-1} \sd_{b_s,\dm_s} \tz_{b_s,\dm_s} v.
\end{equation}
\item[(iii)] The filter bank $\{a\sampled \dm; b_1 \sampled \dm_1, \ldots, b_s\sampled \dm_s\}$ is \emph{a tight framelet filter bank with mixed sampling factors}, that is, the following perfect reconstruction conditions hold:
\be \label{fb:pr:1}
|\wh{a}(\xi)|^2+|\wh{b_1}(\xi)|^2+ \cdots+|\wh{b_s}(\xi)|^2=1, \qquad a.e.\, \xi\in \dR
\ee
and
\be \label{fb:pr:0}
\chi_{\dm^{-\tp}\dZ}(\omega)\wh{a}(\xi) \ol{\wh{a}(\xi+2\pi\omega)}+
\sum_{\ell=1}^s \chi_{\dm_\ell^{-\tp}\dZ}(\omega) \wh{b_\ell}(\xi) \ol{\wh{b_\ell}(\xi+2\pi\omega)}=0,
\ee
for almost every $\xi\in \dR$ and for all $\omega \in [\Omega_{\dm}\cup \cup_{\ell=1}^s \Omega_{\dm_\ell}]\backslash\{0\}$, where $\Omega_{\dm_\ell}:=(\dm_\ell^{-\tp} \dZ) \cap [0,1)^d$ and
$\chi_{\dm_\ell^{-\tp}\dZ}(\omega)=1$ if $\omega\in \dm_\ell^{-\tp}\dZ$ and $\chi_{\dm_\ell^{-\tp}\dZ}(\omega)=0$ if $\omega\not \in \dm_\ell^{-\tp}\dZ$.
\end{enumerate}
\end{theorem}

\begin{proof} The equivalence between item (i) and item (ii) is obvious.
By \eqref{def:sd:tz:F}, we see that the Fourier series of the sequence $\sd_{b_\ell,\dm_\ell} \tz_{b_\ell,\dm_\ell} v$ is
\[
|\det(\dm_\ell)|\sum_{\omega_\ell\in \Omega_{\dm_\ell}}\wh{v}(\xi+2\pi \omega_\ell)\wh{b_\ell}(\xi)\ol{\wh{b_\ell}(\xi+2\pi \omega_\ell)}.
\]
Consequently, we see that \eqref{pr} holds if and only if
\begin{align*}
\wh{v}(\xi)&=\sum_{\omega_0\in \Omega_{\dm}}\wh{v}(\xi+2\pi \omega_0)\wh{a}(\xi)\ol{\wh{a}(\xi+2\pi \omega_0)}+\sum_{\ell=1}^s
\sum_{\omega_\ell\in \Omega_{\dm_\ell}}\wh{v}(\xi+2\pi \omega_\ell)\wh{b_\ell}(\xi)\ol{\wh{b_\ell}(\xi+2\pi \omega_\ell)}\\
&=\sum_{\omega\in \Omega_{\dm}\cup \cup_{\ell=1}^s \Omega_{\dm_\ell}}
\wh{v}(\xi+2\pi \omega)\left( \chi_{\dm^{-\tp}\dZ}(\omega)\wh{a}(\xi) \ol{\wh{a}(\xi+2\pi\omega)}+
\sum_{\ell=1}^s \chi_{\dm_\ell^{-\tp}\dZ}(\omega) \wh{b_\ell}(\xi) \ol{\wh{b_\ell}(\xi+2\pi\omega)}\right).
\end{align*}
Now using the above identity and employing a similar argument as in the proof of \cite[Theorem~2.1]{Han:mmnp:2013}, we can deduce that item (ii) is equivalent to item (iii).
\end{proof}

\subsection{Discrete affine systems of tight framelet filter banks with mixed sampling factors}

To understand the performance and properties of the $J$-level fast framelet transform using a tight framelet filter bank $\{a \sampled \dm; b_1 \sampled \dm_1,\ldots, b_s\sampled \dm_s\}$, as pointed out in \cite{Han:mmnp:2013}, it is very important to look at the $J$-level discrete affine systems associated with $\{a \sampled \dm; b_1 \sampled \dm_1,\ldots, b_s\sampled \dm_s\}$.

We now generalize the discrete affine systems in \cite[Section~4.3]{Han:mmnp:2013} to a $d$-dimensional tight framelet filter bank $\{a \sampled \dm; b_1 \sampled \dm_1,\ldots, b_s\sampled \dm_s\}$ with mixed sampling factors.
Let $a,b_1,\ldots,b_s\in \dlp{1}$. Note that $\dlp{1}\subseteq \dlp{2}$ and
$\dlp{2}$ is a Hilbert space equipped with the inner product $\la u,v\ra:=\sum_{k\in\dZ} u(k)\ol{v(k)}$ for $u,v\in \dlp{2}$.
Following \cite{Han:mmnp:2013}, we define the multilevel filters $a_j$ and $b_{\ell,j}$ with $j\in \N$ and $\ell=1,\ldots,s$ by
\be \label{aj}
\wh{a_j}(\xi):= \wh{a}(\xi)\wh{a}(\dm^\tp\xi)\cdots\wh{a}((\dm^\tp)^{j-2}\xi)\wh{a}((\dm^\tp)^{j-1}\xi)\\
\ee
and
\be \label{blj}
\wh{b_{\ell,j}}(\xi):=\wh{a_{j-1}}(\xi)\wh{b_\ell}((\dm^\tp)^{j-1}\xi)
=\wh{a}(\xi)\wh{a}(\dm^\tp\xi)\cdots\wh{a}((\dm^\tp)^{j-2}\xi)\wh{b_\ell}((\dm^\tp)^{j-1}\xi).
\ee
In particular, $a_1=a$ and $b_{\ell,1}=b_\ell$.
We shall also use the convention $a_0=\td$, where $\td$ is the Dirac/Kronecker sequence on $\dZ$ given by
\[
\td(0)=1 \quad \mbox{and}\quad \td(k)=0,\qquad \forall\, k\in \dZ \bs \{0\}.
\]
Since $a,b_1,\ldots,b_s\in \dlp{1}$, it is straightforward to see that all $a_j, b_{\ell,j}$ are well-defined filters in  $\dlp{1}\subseteq \dlp{2}$. For $j\in \N$ and $k\in \dZ$, we define
\be \label{abjk}
a_{j;k}:=|\det(\dm)|^{j/2} a_j(\cdot-\dm^j k), \qquad
b_{\ell,j;k}:=|\det(\dm)|^{(j-1)/2}|\det(\dm_\ell)|^{1/2} b_{\ell,j}(\cdot-\dm^{j-1}\dm_\ell k).
\ee
The \emph{$J$-level discrete affine system} associated with the filter bank $\{a\sampled \dm; b_1 \sampled \dm_1, \ldots, b_s \sampled \dm_s\}$ is defined by
\be \label{das}
\das_J(\{a\sampled \dm; b_1 \sampled \dm_1, \ldots, b_s \sampled \dm_s\}):=
\{ a_{J;k} \setsp k\in \dZ\} \cup \{ b_{\ell,j;k} \setsp k\in \dZ, \ell=1, \ldots, s, j=1,\ldots, J\}.
\ee
By a similar argument as in \cite[Section~4.3]{Han:mmnp:2013} (also see Theorem~\ref{thm:das} below), under the framework of the Hilbert space $l_2(\dZ)$,
we see that the $J$-level fast framelet transform using the tight framelet filter bank $\{a \sampled \dm; b_1 \sampled \dm_1,\ldots, b_s\sampled \dm_s\}$ is exactly to compute the following representation:
\be \label{das:ffrt}
v=\sum_{u\in \das_J(\{a\sampled \dm; b_1 \sampled \dm_1, \ldots, b_s \sampled \dm_s\})} \la v, u\ra u=
\sum_{k\in \dZ} \la v, a_{J;k}\ra a_{J;k}+\sum_{j=1}^J \sum_{\ell=1}^s \sum_{k\in \dZ} \la v, b_{\ell,j;k}\ra b_{\ell,j; k}, \qquad \forall\, v\in l_2(\dZ),
\ee
where the series converges unconditionally in $l_2(\dZ)$. More precisely, as we shall see later, $v_J(k)=\la v_0, a_{J;k}\ra$ and $w_{\ell,j}(k)=\la v_0,b_{\ell,j;k}\ra$ for all $j=1,\ldots,J$ and $k\in \dZ$.

Following the general theory developed in \cite{Han:mmnp:2013}, we have the following result.

\begin{theorem}\label{thm:das}
Let $a,b_1,\ldots,b_s\in \dlp{1}$ and $\dm, \dm_1, \ldots, \dm_s$ be $d\times d$ invertible integer matrices. For $J\in \N$, define $\das_J(\{a \sampled \dm; b_1 \sampled \dm_1,\ldots, b_s\sampled \dm_s\})$ as in \eqref{das} with $a_j$ and $b_{\ell,j}$ being given in \eqref{aj} and \eqref{blj}, respectively. Then the following statements are equivalent:
\begin{enumerate}
\item[(1)] $\{a \sampled \dm; b_1 \sampled \dm_1,\ldots, b_s\sampled \dm_s\}$ is a tight framelet filter bank with mixed sampling factors.
\item[(2)] The following identity holds:
\be \label{das:level1}
v=\sum_{k\in \dZ} \la v, a_{1;k}\ra a_{1;k}+\sum_{\ell=1}^s \sum_{k\in \dZ} \la v, b_{\ell,1;k}\ra b_{\ell,1;k},\qquad \forall\, v\in \dlp{2}.
\ee
\item[(3)] $\das_1(\{a \sampled \dm; b_1 \sampled \dm_1,\ldots, b_s\sampled \dm_s\})$ is a (normalized) tight frame for $\dlp{2}$, that is,
\be \label{das:frame:1}
\|v\|_{\dlp{2}}^2=\sum_{k\in \dZ} |\la v, a_{1;k}\ra|^2+\sum_{\ell=1}^s \sum_{k\in \dZ} |\la v, b_{\ell,1;k}\ra|^2, \qquad \forall\; v\in \dlp{2}.
\ee
\item[(4)] For every $j\in \N$, the following identity holds:
\be \label{cascade}
\sum_{k\in \dZ} \la v, a_{j-1;k}\ra a_{j-1;k}=\sum_{k\in \dZ} \la v, a_{j;k}\ra a_{j;k}+\sum_{\ell=1}^s \sum_{k\in \dZ} \la v, b_{\ell,j;k}\ra b_{\ell,j;k},\qquad \forall\, v\in l_2(\dZ),
\ee
where by convention $a_{0}:=\td$ and $a_{0;k}:=\td(\cdot-k)$ for $k\in \dZ$.

\item[(5)] For every $J\in \N$, the identity in \eqref{das:ffrt} holds.
\item[(6)] For every $J\in \N$,  $\das_J(\{a \sampled \dm; b_1 \sampled \dm_1,\ldots, b_s\sampled \dm_s\})$ is a (normalized) tight frame for $\dlp{2}$, that is,
\be \label{das:frame:J}
\|v\|_{\dlp{2}}^2=\sum_{k\in \dZ} |\la v, a_{J;k}\ra|^2+\sum_{j=1}^J\sum_{\ell=1}^s \sum_{k\in \dZ} |\la v, b_{\ell,j;k}\ra|^2, \qquad \forall\; v\in \dlp{2}.
\ee
\end{enumerate}
\end{theorem}

\begin{proof} These claims have been established in \cite{Han:mmnp:2013} for the case $\dm_1=\cdots=\dm_s=\dm$. Using the same idea as in \cite{Han:mmnp:2013},
here we only present a sketch of a proof. Plugging $v=\td(\cdot-n)$ with all $n\in \dZ$ into \eqref{das:level1}, we observe by calculation that the resulting equations in \eqref{das:level1} with $v=\td(\cdot-n)$
are simply the spatial domain version of the conditions in \eqref{fb:pr:1} and \eqref{fb:pr:0} in the frequency domain. Hence, (1)$\iff$(2).

(2)$\imply$(3) is trivial. (3)$\imply$(2) is a direct application of the polarization identity to \eqref{das:frame:1}. Hence, (2)$\iff$(3).

(4)$\imply$(2) is obvious since it follows from the convention $a_0=\td$ that $\sum_{k\in \dZ}\la v, a_{0;k}\ra a_{0;k}=\sum_{k\in \dZ} v(k) \td(\cdot-k)=v$.
We now prove (2)$\imply$(4). By the definition of $b_{\ell,j}$ in \eqref{blj} and $b_{\ell,1}=b_\ell$,
\begin{align*}
b_{\ell,j}&=a_{j-1}*(b_\ell \us \dm^{j-1})=a_{j-1}*(b_{\ell,1} \us \dm^{j-1})\\
&= \sum_{n\in \dZ}a_{j-1}(\cdot-n)(b_{\ell,1} \us \dm^{j-1})(n)
=\sum_{m\in \dZ} a_{j-1}(\cdot-\dm^{j-1} m) b_{\ell,1}(m).
\end{align*}
Therefore, by the definition of $b_{\ell,j;k}$ in \eqref{abjk},
\begin{align*}
b_{\ell,j;k}&=
|\det(\dm)|^{(j-1)/2} |\det(\dm_\ell)|^{1/2} b_{\ell,j}(\cdot-\dm^{j-1}\dm_\ell k)\\
&=|\det(\dm)|^{(j-1)/2} |\det(\dm_\ell)|^{1/2} \sum_{m\in \dZ} a_{j-1}(\cdot-\dm^{j-1} \dm_\ell k-\dm^{j-1} m)b_{\ell,1}(m)\\
&=|\det(\dm)|^{(j-1)/2} |\det(\dm_\ell)|^{1/2} \sum_{m\in \dZ} a_{j-1}(\cdot-\dm^{j-1} m)b_{\ell,1}(m-\dm_\ell k)\\
&=\sum_{m\in \dZ} a_{j-1;m} b_{\ell,1;k}(m).
\end{align*}
Consequently, we proved
\be \label{vbljk}
\la v, b_{\ell,j;k}\ra
=\sum_{m\in \dZ} \la v, a_{j-1;m}\ra \ol{b_{\ell,1;k}(m)}=\la \la v, a_{j-1;\cdot}\ra, b_{\ell,1;k}(\cdot)\ra.
\ee
We now deduce from the above two identities that
\[
\sum_{k\in \dZ}\la v, b_{\ell,j;k}\ra b_{\ell,j;k}
=\sum_{m\in \dZ} a_{j-1;m} \left(\sum_{k\in \dZ} \la \la v, a_{j-1;\cdot}\ra, b_{\ell,1;k} \ra b_{\ell,1;k}(m)\right).
\]
The same argument can be applied to $a_{j;k}$ and the above identity still holds by replacing $b_{\ell,j;k}$ and $b_{\ell,1;k}$ with $a_{j;k}$ and $a_{1;k}$, respectively.
Therefore,
\begin{align*}
&\sum_{k\in \dZ} \la v, a_{j;k}\ra a_{j; k}+\sum_{\ell=1}^s \sum_{k\in \dZ} \la v, b_{\ell,j;k}\ra b_{\ell,j;k}\\
&=\sum_{m\in \dZ} a_{j-1;m} \left(\sum_{k\in \dZ} \la \la v, a_{j-1;\cdot}\ra, a_{1;k}\ra a_{1;k}(m)+\sum_{\ell=1}^s \sum_{k\in \dZ} \la \la v, a_{j-1;\cdot}\ra, b_{\ell,1;k}\ra b_{\ell,1;k}(m)\right)\\
&=\sum_{m\in \dZ}  \la v, a_{j-1;m}\ra a_{j-1;m},
\end{align*}
where we used \eqref{das:level1}, i.e., item (2), in the last identity.
This proves (2)$\imply$(4).

(4)$\imply$(5) is obvious. Conversely, considering the differences between $J=j$ and $J=j-1$ in \eqref{das:ffrt}, we see that (5)$\imply$(4).
The equivalence between item (5) and item (6) is straightforward and is similar to the equivalence between item (2) and item (3).
\end{proof}

We now show that
the coefficients in the representation in \eqref{das:ffrt} using a $J$-level discrete affine system can be exactly computed through the $J$-level fast framelet decomposition in \eqref{def:wa}.
Since $\tz_{u,\dm} v=|\det(\dm)|(v*u^\star)\ds \dm$ and $\wh{a_{j-1}}(\xi)=\wh{a}(\xi)\cdots \wh{a}((\dm^\tp)^{j-2}\xi)$, by \cite[Lemma~4.3]{Han:mmnp:2013}, we have
\begin{align*}
\la v, a_{j-1;k}\ra&=|\det(\dm)|^{(j-1)/2}\la v, a_{j-1}(\cdot-\dm^{j-1}k) \ra=|\det(\dm)|^{(1-j)/2} [\tz_{a_{j-1}, \dm^{j-1}} v](k)\\
&=
|\det(\dm)|^{(1-j)/2} [\tz_{a, \dm}^{j-1} v](k)=v_{j-1}(k),
\end{align*}
where $v_{j-1}$ is exactly the same sequence as obtained in the fast framelet decomposition in \eqref{def:wa} with $v_0:=v$. Similarly, by \eqref{vbljk} and the above identity, we have
\begin{align*}
\la v,b_{\ell,j;k}\ra&=\la \la v,a_{j-1;\cdot}\ra, b_{\ell,1;k}\ra=
|\det(\dm_\ell)|^{1/2} \la v_{j-1}, b_\ell(\cdot-\dm_\ell k)\ra\\
&=|\det(\dm_\ell)|^{1/2} \sum_{m\in \dZ} v_{j-1}(m)\ol{b_\ell(m-\dm_\ell k)}=
|\det(\dm_\ell)|^{-1/2} [\tz_{b_\ell,\dm_\ell} v_{j-1}](k)=w_{\ell,j}(k).
\end{align*}
This establishes the connection between the representation in \eqref{das:ffrt} under the $J$-level discrete affine system and the $J$-level fast/discrete framelet transform in \eqref{def:wa} and \eqref{def:ws}.

\subsection{Connections to tight framelets in $L_2(\R^d)$}

Following the general theory on frequency-based framelets in \cite{Han:acha:2010,Han:acha:2012}, we now discuss the natural connections of a tight framelet filter bank $\{a \sampled \dm; b_1 \sampled \dm_1, \ldots, b_s\sampled \dm_s\}$ with a tight framelet in $L_2(\R^d)$.

For a function $f: \R^d\rightarrow \C$ and a $d\times d$ real-valued matrix $U$, following \cite{Han:acha:2012}, we shall adopt the following notation:
\[
f_{U; k,n}(x):=f_{\lb U; k,n\rb }(x):=\lb U; k,n\rb f(x):=|\det(U)|^{1/2}e^{-in\cdot Ux} f(Ux-k), \qquad x,k,n\in \dR.
\]
In particular, we define $f_{U;k}:=f_{U;k,0}=|\det U|^{1/2} f(U\cdot-k)$. For $f\in \dLp{1}$, its Fourier transform is defined to be $\wh{f}(\xi):=\int_{\dR} f(x) e^{-ix\cdot \xi} dx$ for $\xi\in \dR$. Note that $\wh{f_{U;k}}=\wh{f}_{U^{-\tp};0,k}$.

The following result is based on the general theory developed in \cite{Han:acha:2010,Han:acha:2012} on frequency-based framelets.

\begin{theorem}\label{thm:tf}
Let $a,b_1,\ldots, b_s\in \dlp{1}$ and $\dm,\dm_1,\ldots,\dm_s$ be $d\times d$ invertible integer matrices.
Suppose that all the eigenvalues of $\dm$ are greater than one in modulus and there exist positive numbers $\gep, C, \tau$ such that
%
$|1-\wh{a}(\xi)|\le C\|\xi\|^\tau$ for all $\xi\in [-\gep,\gep]^d$.
%
Define
\be \label{phi:psi}
\wh{\phi}(\xi):=\prod_{j=1}^\infty \wh{a}((\dm^\tp)^{-j}\xi) \quad
\mbox{and}\quad
\wh{\psi^\ell}(\xi):=\wh{b_\ell}(\dm^{-\tp}\xi)\wh{\phi}(\dm^{-\tp}\xi), \qquad \xi\in \dR, \ell=1, \ldots,s.
\ee
If $\{a \sampled \dm; b_1 \sampled \dm_1,\ldots, b_s \sampled \dm_s\}$ is a tight framelet filter bank, then $\{\phi \sampled \dm; \psi^1 \sampled \dm_1,\ldots, \psi^s \sampled \dm_s\}$ is a tight framelet in $\dLp{2}$, that is, $\phi,\psi^1,\ldots, \psi^s\in \dLp{2}$ and $\AS_0(\{\phi \sampled \dm; \psi^1 \sampled \dm_1,\ldots, \psi^s \sampled \dm_s\})$ is a (normalized) tight frame for $\dLp{2}$:
\be \label{tf:L2}
\|f\|_{\dLp{2}}^2=\sum_{k\in \dZ} |\la f, \phi(\cdot-k)\ra|^2+
\sum_{j=0}^\infty \sum_{\ell=1}^s \sum_{k\in \dZ} |\la f, |\det(\dm^{-1}\dm_\ell)|^{1/2}\psi^\ell_{\dm^{j};\dm^{-1}\dm_\ell k}\ra|^2,
\ee
for all $f\in \dLp{2}$, where
\be \label{AS0}
\begin{split}
\AS_0(\{\phi \sampled \dm; &\psi^1 \sampled \dm_1,\ldots, \psi^s \sampled \dm_s\}):=\{\phi(\cdot-k) \setsp k\in \dZ\}\\
&\cup \{ |\det(\dm^{-1}\dm_\ell)|^{1/2}\psi^\ell_{\dm^{j};\dm^{-1}\dm_\ell k} \setsp k\in \dZ, \ell=1,\ldots,s, j\in \N\cup\{0\}\}.
\end{split}
\ee
The converse direction also holds provided in addition that $\sum_{k\in \dZ} |\wh{\phi}(\xi+2\pi k)|^2\ne 0$ for almost every $\xi\in \dR$.
\end{theorem}

\begin{proof} By the same argument as in \cite[Theorem~13]{Han:acha:2012} and \cite[Theorem~6]{Han:acha:2010},
$\phi,\psi^1,\ldots, \psi^s\in \dLp{2}$ and \eqref{tf:L2} holds for all $f\in \dLp{2}$ if and only if
\be \label{to1}
\lim_{j\to +\infty} \sum_{k\in \dZ} |\la f, \phi_{\dm^j;k}\ra|^2=\|f\|_{\dLp{2}}^2
\ee
and
\be \label{twolevel}
\sum_{k\in \dZ} |\la f, \phi(\cdot-k)\ra|^2+\sum_{\ell=1}^s \sum_{k\in \dZ} |\la f, |\det(\dm^{-1}\dm_\ell)|^{1/2} \psi^\ell(\cdot- \dm^{-1}\dm_\ell k)\ra|^2=\sum_{k\in \dZ} |\la f,\phi_{\dm^{-1};k}\ra|^2
\ee
for all $f\in \dLp{2}$ such that $\wh{f}$ is a compactly supported $C^\infty$ function.

By our assumption on $\dm$ and $\wh{a}$, we see that $\wh{\phi}$ is a well-defined bounded function.  By a similar argument as in \cite[Lemma~4]{Han:acha:2010}, we see that \eqref{to1} is satisfied, since $\lim_{j\to +\infty}\wh{\phi}((\dm^\tp)^{-j}\xi)=1$.

Define $\dn:=\dm^{-\tp}$ and $\dn_\ell:=\dm_\ell^{-\tp}$.
Note that $\psi^\ell(\cdot-\dm^{-1}\dm_\ell k)=\eta^\ell(\dm_\ell^{-1}\dm\cdot-k)$ with $\eta^\ell:=\psi^\ell(\dm^{-1}\dm_\ell\cdot)$ and $\wh{\eta^\ell}(\xi)=|\det(\dm_\ell^{-1}\dm)|\wh{\psi^\ell}(\dn^{-1}\dn_\ell\xi)$.
By \cite[Lemma~10]{Han:acha:2012},
we have
\begin{align*}
&\sum_{k\in \dZ} |\la f, |\det(\dm^{-1}\dm_\ell)|^{1/2} \psi^\ell(\cdot- \dm^{-1}\dm_\ell k)\ra|^2 =|\det(\dm^{-1}\dm_\ell)|^2
\sum_{k\in \dZ} |\la f, \eta^\ell_{\dm_\ell^{-1} \dm; k}\ra|^2\\
&\quad =(2\pi)^{-2d} |\det(\dm^{-1} \dm_\ell)|^2 \sum_{k\in \dZ} |\la \wh{f},\wh{\eta^\ell}_{\dn_\ell^{-1}\dn;0,k}\ra|^2\\
&\quad =(2\pi)^{-d}  \int_{\dR}
\sum_{k\in \dZ} \wh{f}(\xi)\ol{\wh{f}(\xi+2\pi \dn^{-1}\dn_\ell k)}
\ol{\wh{\psi^\ell}(\xi)} \wh{\psi^\ell}(\xi+2\pi \dn^{-1}\dn_\ell k)d\xi\\
&\quad =(2\pi)^{-d}  \int_{\dR}
\sum_{k\in \dZ} \wh{f}(\xi)\ol{\wh{f}(\xi+2\pi \dn^{-1}\dn_\ell k)}
\ol{b_\ell(\dn\xi)} \wh{b_\ell}(\dn\xi+2\pi \dn_\ell k)
\ol{\wh{\phi}(\dn \xi)} \wh{\phi}(\dn \xi+2\pi \dn_\ell k)d\xi\\
&\quad =(2\pi)^{-d}  \int_{\dR} \wh{f}(\xi)  \ol{\wh{\phi}(\dn \xi)} \sum_{\omega_\ell\in \Omega_\ell}
 \ol{b_\ell(\dn\xi)} \wh{b_\ell}(\dn\xi+2\pi \omega_\ell)
\sum_{k\in \dZ} \ol{\wh{f}(\xi+2\pi \dn^{-1}\omega_\ell+2\pi \dn^{-1} k)}
\wh{\phi}(\dn \xi+2\pi \omega_\ell+2\pi k)d\xi,
\end{align*}
where we used \eqref{phi:psi} in the last second identity and the fact that $\dZ=\dm_\ell^\tp \Omega_\ell+\dm_\ell^\tp \dZ$.
Similarly, by \cite[Lemma~10]{Han:acha:2012} we have
\begin{align*}
&\sum_{k\in \dZ} |\la f, \phi(\cdot-k)\ra|^2=(2\pi)^{-d} \int_{\dR} \sum_{k\in \dZ}\wh{f}(\xi)\ol{\wh{f}(\xi+2\pi k)} \ol{\wh{a}(\dn\xi)} \wh{a}(\dn \xi+2\pi \dn k) \ol{\wh{\phi}(\dn \xi)} \wh{\phi}(\dn \xi+2\pi \dn k)d\xi\\
&\qquad =(2\pi)^{-d} \int_{\dR} \wh{f}(\xi) \ol{\wh{\phi}(\dn \xi)} \sum_{\omega_0\in \Omega_{\dm}}
 \ol{\wh{a}(\dn\xi)} \wh{a}(\dn \xi+2\pi \omega_0) \sum_{k\in \dZ} \ol{\wh{f}(\xi+2\pi \dn^{-1}\omega_0+ 2\pi \dn^{-1} k)}  \wh{\phi}(\dn \xi+2\pi \omega_0+2\pi k)d\xi
\end{align*}
and
\[
\sum_{k\in \dZ} |\la f, \phi_{\dm;k}\ra|^2=(2\pi)^{-d} \int_{\dR}
\wh{f}(\xi) \ol{\wh{\phi}(\dn\xi)} \sum_{k\in \dZ} \ol{\wh{f}(\xi+2\pi \dn^{-1} k)} \wh{\phi}(\dn\xi+2\pi k)d\xi.
\]
By a similar argument as in \cite[Lemma~5]{Han:acha:2010}, we can conclude that \eqref{twolevel} holds if and only if
\be \label{twolevel:phi}
\begin{split}
\ol{\wh{\phi}(\xi)}\wh{\phi}(\xi+2\pi \omega+2\pi k) &\left(
\chi_{\dm^{-\tp}\dZ}(\omega)\ol{\wh{a}(\xi)} \wh{a}(\xi+2\pi\omega)
+\sum_{\ell=1}^s \chi_{\dm_\ell^{-\tp}\dZ}(\omega) \ol{\wh{b_\ell}(\xi)} \wh{b_\ell}(\xi+2\pi\omega)\right)\\
&=\td(\omega) \ol{\wh{\phi}(\xi)}\wh{\phi}(\xi+2\pi k),\qquad a.e.\, \xi\in \dR
\end{split}
\ee
for all $\omega\in \Omega_\dm\cup \cup_{\ell=1}^s \Omega_{\ell}$ and for all $k\in \dZ$.
If $\{a \sampled \dm; b_1 \sampled \dm_1,\ldots, b_s \sampled \dm_s\}$ is a tight framelet filter bank, by \eqref{fb:pr:1} and \eqref{fb:pr:0}, it is obvious that \eqref{twolevel:phi} is satisfied and therefore,
$\{\phi \sampled \dm; \psi^1 \sampled \dm_1,\ldots, \psi^s \sampled \dm_s\}$ is a tight framelet for $\dLp{2}$.

If $\sum_{k\in \dZ} |\wh{\phi}(\xi+2\pi k)|^2\ne 0$ for almost every $\xi\in \dR$, then it is easy to deduce that \eqref{twolevel:phi} is equivalent to \eqref{fb:pr:1} and \eqref{fb:pr:0}. This proves the converse direction.
\end{proof}

Since $\dm^{-1}\dm_\ell \dZ =\dZ$ may not hold any more for all $\ell=1,\ldots,s$, the system $\AS_0(\{\phi \sampled \dm; \psi^1 \sampled \dm_1,\ldots, \psi^s \sampled \dm_s\})$ in \eqref{AS0} is not covered by the traditional theory of wavelet analysis.


\section{Directional Tensor Product Complex Tight Framelets with Low Redundancy}
\label{sec:tpctf}

In this section we first briefly recall the directional tensor product complex tight framelets from \cite{Han:mmnp:2013,HanZhao:siims:2014}.
Built on the results on tight framelet filter banks with mixed sampling factors in Section~\ref{sec:tffb:ms}, we shall provide the details on our proposed directional tensor product complex tight framelet filter bank $\tpctf^{\reduced}_6$ with low redundancy rate.

\subsection{Tensor product complex tight framelets and their redundancy rates}

For $c_L<c_R$ and positive numbers $\gep_L, \gep_R$ satisfying $\gep_L+\gep_R\le c_R-c_L$, we define a bump function $\chi_{[c_L, c_R]; \gep_L, \gep_R}$ on $\R$ (\cite{Han:acha:1997,Han:mmnp:2013,HanZhao:siims:2014}) by
\be \label{bump:func}
\chi_{[c_L, c_R]; \gep_L, \gep_R}(\xi):=
\begin{cases} 0, \quad &\xi\le c_L-\gep_L \; \mbox{or}\; \xi \ge c_R+\gep_R, \\
\cos\big(\tfrac{\pi(c_L+\gep_L-\xi)}{4\gep_L}\big), \quad &c_L-\gep_L<\xi<c_L+\gep_L,\\
1, \quad &c_L+\gep_L\le \xi\le c_R-\gep_R,\\
\cos\big(\tfrac{\pi(\xi-c_R+\gep_R)}{4\gep_R}\big),
\quad &c_R-\gep_R<\xi<c_R+\gep_R.
\end{cases}
\ee
Note that $\chi_{[c_L, c_R]; \gep_L, \gep_R}$ is a continuous function supported on $[c_L-\gep_L,c_R+\gep_R]$.

Let $s\in \N$ and $0<c_1<c_2<\cdots<c_{s+1}:=\pi$ and $\gep_0,\gep_1,\ldots, \gep_{s+1}$ be positive real numbers satisfying
\[
\gep_0+\gep_1\le c_1\le \tfrac{\pi}{2}-\gep_1
\quad \mbox{and}\quad
\gep_\ell+\gep_{\ell+1}\le
c_{\ell+1}-c_\ell \le \pi-\gep_{\ell}-\gep_{\ell+1}, \qquad\forall\; \ell=1,\ldots,s.
\]
A real-valued low-pass filter $a$ and $2s$ complex-valued high-pass filters $b^{p}_1, \ldots, b^{p}_s, b^{n}_1, \ldots, b^{n}_s$ are defined through their Fourier series on the basic interval $[-\pi, \pi)$ as follows:
\be \label{ctf}
\wh{a}:=\chi_{[-c_1, c_1]; \gep_1, \gep_1}, \quad
\wh{b^{p}_{\ell}}:=\chi_{[c_\ell,c_{\ell+1}]; \gep_\ell, \gep_{\ell+1}}, \quad
\wh{b_{\ell}^{n}}:=\ol{\wh{b_{\ell}^{p}}(-\cdot)}, \qquad \ell=1, \ldots, s.
\ee
Then  $\ctf_{2s+1}:=\{a; b^{p}_1, \ldots, b^{p}_s, b^{n}_1, \ldots, b^{n}_s\}$ is a (one-dimensional dyadic) tight framelet filter bank.
The tensor product complex tight framelet filter bank $\tpctf_{2s+1}$ for dimension $d$ is simply
\[
\tpctf_{2s+1}:=\otimes^d \ctf_{2s+1}=\otimes^d \{a; b^{p}_1, \ldots, b^{p}_s, b^{n}_1, \ldots, b^{n}_s\}.
\]
This tight framelet filter bank $\tpctf_{2s+1}$ has one real-valued low-pass filter $\otimes^d a$ and $(2s+1)^d-1$ complex-valued high-pass filters. This family of tensor product complex tight framelets has been introduced in \cite{Han:mmnp:2013}.

To further improve the directionality of $\tpctf_{2s+1}$, another closely related family of tensor product complex tight framelet filter banks $\tpctf_{2s+2}$ has been introduced in \cite{HanZhao:siims:2014}. Define filters $a, b^{p}_1, \ldots, b^{p}_s, b^{n}_1, \ldots, b^{n}_s$ as in \eqref{ctf}. Define two auxiliary complex-valued filters $a^p, a^n$ by
\begin{equation}\label{apn}
\wh{a^{p}}:=\chi_{[0, c_1]; \gep_0, \gep_1}, \qquad
\wh{a^{n}}:=\ol{\wh{a^p}(-\cdot)}.
\end{equation}
Then $\ctf_{2s+2}:=\{a^p, a^n; b^{p}_1, \ldots, b^{p}_s, b^{n}_1, \ldots, b^{n}_s\}$ is also a (one-dimensional dyadic) tight framelet filter bank. Now the tensor product complex tight framelet filter bank $\tpctf_{2s+2}$ for dimension $d$ is defined to be
\[
\tpctf_{2s+2}:=\{\otimes^d a; \tpctf\mbox{-HP}_{2s+2}\},
\]
where $\tpctf\mbox{-HP}_{2s+2}$ consists of total
$(2s+2)^d-2^d$ complex-valued high-pass filters given by
\[
\Big(\otimes^d \{a^p, a^n, b^{p}_1, \ldots, b^{p}_s, b^{n}_1, \ldots, b^{n}_s\}\Big)\bs \Big(\otimes^d \{a^p, a^n\}\Big).
\]
The sampling matrices/factors for all tensor product complex tight framelet filter banks $\tpctf_m$ with $m\ge 3$ are $2I_d$.
See \cite{Han:mmnp:2013,HanMoZhao:2013,HanZhao:siims:2014,SHB:2014} for detailed discussions on tensor product complex tight framelets and their applications to image processing.

We now discuss the redundancy rates of $\tpctf_m$ with $m\ge 3$.
Note that $\wh{b_{\ell}^{n}}=\ol{\wh{b_{\ell}^{p}}(-\cdot)}$ is equivalent to $b_{\ell}^{n}=\ol{b_{\ell}^{p}}$, that is, $b_{\ell}^{n}(k)=\ol{b_{\ell}^{p}(k)}$ for all $k\in \Z$.
Therefore, by the last identities in \eqref{ctf} and \eqref{apn}, we can always rewrite the tight framelet filter bank $\tpctf_{m}$ as
\be \label{tpctf:half}
\tpctf_{m}=\{ \otimes^d a;  u, \ol{u}\;\mbox{with}\;  u\in \tpctf\operatorname{-CHP}_{m}\},
\ee
where $\tpctf\operatorname{-CHP}_{m}$ is a subset of $\tpctf_{m}$ and has exactly $n_m$ filters, where $n_m:=\frac{m^d-1}{2}$ for odd integers $m$  and $n_m:=\frac{m^d-2^d}{2}$ for even integers $m$.
For a complex-valued filter $u:\dZ \rightarrow \C$, we can uniquely write
$u=\re(u)+i\im(u)$, where $\re(u)$ and $\im(u)$ are two real-valued filters defined by $\re(u)(k):=\re(u(k))$ and $\im(u)(k):=\im(u(k))$ for all $k\in \dZ$.
Due to the identity in \eqref{tpctf:half}, we observe that the complex-valued tight framelet filter bank $\tpctf_{m}$ is essentially equivalent to the following real-valued tight framelet filter bank:
\be \label{tpctfm:real}
\{\otimes^d a\}\cup\{ \sqrt{2}\re(u), \sqrt{2}\im(u) \setsp  u\in  \tpctf\operatorname{-CHP}_{m}\},
\ee
which has one real-valued low-pass filter and $2n_m$ real-valued high-pass filters. Therefore,
since the sampling matrices are $2I_d$ with determinant $2^d$, the redundancy rate of $\tpctf_m$ in dimension $d$ is no more than
\[
\frac{2n_m}{2^d}\sum_{j=0}^\infty \frac{1}{2^{dj}}=\frac{2n_m}{2^d-1}=\begin{cases}
\frac{m^d-1}{2^d-1}, &\text{if $m$ is an odd integer,}\\
\frac{m^d-2^d}{2^d-1}, &\text{if $m$ is an even integer.}
\end{cases}
\]

\subsection{Directional tensor product complex tight framelets with low redundancy}

Now we are ready to construct directional tensor product complex tight framelets with low redundancy by using large sampling factors for $\tpctf_m$. Though all our arguments in this subsection can be applied to every $\tpctf_m$ with $m\ge 3$, since the directional tensor product complex tight framelet $\tpctf_6$ has been known to have superior performance for image denoising in \cite{HanZhao:siims:2014} and for image inpainting in \cite{SHB:2014}, we shall only concentrate here on the modification of $\tpctf_6$.

As discussed in detail in \cite{Han:mmnp:2013,HanMoZhao:2013,HanZhao:siims:2014}, the directionality of the tensor product complex tight framelets is closely related to the frequency separation property of the high-pass filters in its underlying one-dimensional tight framelet filter bank. More precisely, for a filter $u$, we say that $u$ has \emph{good frequency separation property} if either $\wh{u}(\xi)\approx 0$ for all $\xi\in [-\pi,0]$ or $\wh{u}(\xi)\approx 0$ for all $\xi\in[0,\pi]$. Moreover, we say that a filter $u$ has the \emph{ideal frequency separation property} if either $\wh{u}(\xi)=0$ for all $\xi\in [-\pi,0]$ or $\wh{u}(\xi)=0$ for all $\xi\in[0,\pi]$.

In this subsection, we are interested in building a one-dimensional tight framelet filter bank $\ctf_6^{\reduced}$ (called reduced $\ctf_6$ or $\ctf_6$ down $4$), which consists of one real-valued low-pass filter $a$, two auxiliary complex-valued filters $a^p, a^n$, and four complex-valued high-pass filters $b^{p}_1, b^{p}_2, b^{n}_1, b^{n}_2$ such that
\begin{enumerate}
\item[(1)] $a^n=\ol{a^p}$, $b^{n}_1=\ol{b^{p}_1}$, and $b^{n}_2=\ol{b^{p}_2}$.
\item[(2)] Both $\{a \sampled 2; b^{p}_1\sampled 4, b^{p}_2\sampled 4, b^{n}_1\sampled 4, b^{n}_2\sampled 4\}$ and $\ctf_6^{\reduced}:=\{a^p \sampled 4, a^n \sampled 4; b^{p}_1\sampled 4, b^{p}_2\sampled 4, b^{n}_1\sampled 4, b^{n}_2\sampled 4\}$ are tight framelet filter banks.
\item[(3)] The auxiliary filters $a^p, a^n$ and all the high-pass filters $b^{p}_1, b^{p}_2, b^{n}_1, b^{n}_2$ have good frequency separation property.
\end{enumerate}

See Figure~\ref{fig:ffrt} for an illustration of a one-dimensional multilevel fast framelet transform employing a filter bank $\{a \sampled 2; b_1\sampled 4,\ldots, b_s\sampled 4\}$.

\begin{figure}[ht]
\centering
\pspicture[](0,0)(14,4.6)
\rput(0,1.5){\rnode{input}{\psframebox[framearc=.3]{\tiny input}}}
\rput(0.8,1.5){\rnode{inputjoint}{}}
\rput(1.7,3){\rnode{a}{\psframebox[framearc=.3]{\tiny $\sqrt{2} a^\star$}}}
\rput(1.7,1.82){\rnode{b1}{\psframebox[framearc=.3]{\tiny $\sqrt{4} b_1^\star$}}}
\rput(1.7,0){\rnode{bs}{\psframebox[framearc=.3]{\tiny $\sqrt{4} b_s^\star$}}}
\cnodeput(2.8,3){ds0}{\tiny $\downarrow \! 2$}
\cnodeput(2.8,1.82){ds1}{\tiny $\downarrow \! 4$}
\cnodeput(2.8,0){dss}{\tiny $\downarrow \! 4$}
\rput(3.4,3){\rnode{inputjoint2}{}}
\rput(4.3,4.2){\rnode{a2}{\psframebox[framearc=.3]{\tiny $\sqrt{2} a^\star$}}}
\rput(4.3,3.4){\rnode{b12}{\psframebox[framearc=.3]{\tiny $\sqrt{4} b_1^\star$}}}
\rput(4.3,2.4){\rnode{bs2}{\psframebox[framearc=.3]{\tiny $\sqrt{4} b_s^\star$}}}
\cnodeput(5.5,4.2){ds02}{\tiny $\downarrow \! 2$}
\cnodeput(5.5,3.4){ds12}{\tiny $\downarrow \! 4$}
\cnodeput(5.5,2.4){dss2}{\tiny $\downarrow \! 4$}
\rput(7,4.2){\rnode{proc02}{\psframebox[framearc=.3]{\tiny processing}}}
\rput(7,3.4){\rnode{proc12}{\psframebox[framearc=.3]{\tiny processing}}}
\rput(7,2.4){\rnode{procs2}{\psframebox[framearc=.3]{\tiny processing}}}
\cnodeput(8.5,4.2){us02}{\tiny $\uparrow \! 2$}
\cnodeput(8.5,3.4){us12}{\tiny $\uparrow \! 4$}
\cnodeput(8.5,2.4){uss2}{\tiny $\uparrow \! 4$}
\rput(9.6,4.2){\rnode{ta2}{\psframebox[framearc=.3]{\tiny $\sqrt{2} a$}}}
\rput(9.6,3.4){\rnode{tb12}{\psframebox[framearc=.3]{\tiny $\sqrt{4} b_1$}}}
\rput(9.6,2.4){\rnode{tbs2}{\psframebox[framearc=.3]{\tiny $\sqrt{4} b_s$}}}
\rput(10.3,3){\rnode{outputjoint2}{$\oplus$}}
\ncline{-}{ds0}{inputjoint2}
\ncangle[angleA=90,angleB=180,arm=0.05]{->}{inputjoint2}{a2}
\ncangle[angleA=90,angleB=180,arm=0.05]{->}{inputjoint2}{b12}
\ncangle[angleA=-90,angleB=180,arm=0.05]{->}{inputjoint2}{bs2}
\ncline{->}{a2}{ds02} \ncline{->}{ds02}{proc02}
\ncline{->}{proc02}{us02} \ncline{->}{us02}{ta2}
\ncline{->}{b12}{ds12} \ncline{->}{ds12}{proc12}
\ncline{->}{proc12}{us12} \ncline{->}{us12}{tb12}
\ncline{->}{bs2}{dss2} \ncline{->}{dss2}{procs2}
\ncline{->}{procs2}{uss2} \ncline{->}{uss2}{tbs2}
\ncangle[angleA=0,angleB=90,arm=0.05]{->}{ta2}{outputjoint2}
\ncangle[angleA=0,angleB=90,arm=0.05]{tb12}{outputjoint2}
\ncangle[angleA=0,angleB=-90,arm=0.05]{->}{tbs2}{outputjoint2}
\ncline[nodesep=1pt,linestyle=dotted]{b12}{bs2}
\ncline[nodesep=2pt,linestyle=dotted]{proc12}{procs2}
\ncline[nodesep=1pt,linestyle=dotted]{tb12}{tbs2}
\rput(7,1.82){\rnode{proc1}{\psframebox[framearc=.3]{\tiny processing}}}
\rput(7,0){\rnode{procs}{\psframebox[framearc=.3]{\tiny processing}}}
\cnodeput(11.1,3){us0}{\tiny $\uparrow \! 2$}
\cnodeput(11.1,1.82){us1}{\tiny $\uparrow \! 4$}
\cnodeput(11.1,0){uss}{\tiny $\uparrow \! 4$}
\rput(12.2,3){\rnode{ta}{\psframebox[framearc=.3]{\tiny $\sqrt{2} a$}}}
\rput(12.2,1.82){\rnode{tb1}{\psframebox[framearc=.3]{\tiny $\sqrt{4} b_1$}}}
\rput(12.2,0){\rnode{tbs}{\psframebox[framearc=.3]{\tiny $\sqrt{4} b_s$}}}
\rput(12.9,1.5){\rnode{outputjoint}{$\oplus$}}
\rput(14,1.5){\rnode{output}{\psframebox[framearc=.3]{\tiny output}}}
\ncline{-}{input}{inputjoint}\ncline{->}{outputjoint2}{us0}
\ncangle[angleA=90,angleB=180,arm=0.1]{->}{inputjoint}{a}
\ncangle[angleA=90,angleB=180,arm=0.1]{->}{inputjoint}{b1}
\ncangle[angleA=-90,angleB=180,arm=0.1]{->}{inputjoint}{bs}
\ncline{->}{a}{ds0} \ncline{->}{us0}{ta}
\ncline{->}{b1}{ds1} \ncline{->}{ds1}{proc1}
\ncline{->}{proc1}{us1} \ncline{->}{us1}{tb1}
\ncline{->}{bs}{dss} \ncline{->}{dss}{procs}
\ncline{->}{procs}{uss} \ncline{->}{uss}{tbs}
\ncangle[angleA=0,angleB=90,arm=0.1]{->}{ta}{outputjoint}
\ncangle[angleA=0,angleB=90,arm=0.1]{->}{tb1}{outputjoint}
\ncangle[angleA=0,angleB=-90,arm=0.1]{->}{tbs}{outputjoint}
\ncline{->}{outputjoint}{output}
\ncline[nodesep=14pt,linestyle=dotted]{b1}{bs}
\ncline[nodesep=14pt,linestyle=dotted]{proc1}{procs}
\ncline[nodesep=14pt,linestyle=dotted]{tb1}{tbs}
\endpspicture
\bigskip
\begin{caption}
{Diagram of the one-dimensional two-level discrete framelet transform using a one-dimensional tight framelet filter bank $\{a \sampled 2; b_1\sampled 4,\ldots, b_s\sampled 4\}$. Here each box with a filter inside it means convolution with the filter inside the box. Note that $\tz_{a,2} v=2(v*a^\star) \ds 2$ and $\sd_{a,2} v=2 (v \us 2)*a$, while $\tz_{b_\ell,4} v=4(v*b_\ell^\star) \ds 4$ and $\sd_{b_\ell,4} v=4 (v \us 4)*b_\ell$ for $\ell=1,\ldots,s$. Note that $a^\star$ is the flip-conjugate sequence of $a$ given by $a^\star(k):=\ol{a(-k)}$ for all $k\in \Z$, or equivalent, $\wh{a^\star}(\xi)=\ol{\wh{a}(\xi)}$.
}\label{fig:ffrt}
\end{caption}
\end{figure}

The directionality of the tensor product complex tight framelet $\tpctf_6^{\reduced}$, which we shall introduce later, largely depends on the frequency separation property of all the high-pass filters in the $J$-level discrete affine system $\das_J(\{a \sampled 2; b^{1,p}\sampled 4, b^{2,p}\sampled 4, b^{1,n}\sampled 4, b^{2,n}\sampled 4\})$ as well as the frequency separation property of the two auxiliary filters $a^p$ and $a^n$. For $j\in \N$ and $\ell=1,2$, we define
\begin{align}
&\wh{a_j}(\xi):= \wh{a}(\xi)\wh{a}(2\xi)\cdots\wh{a}(2^{j-2}\xi)\wh{a}(2^{j-1}\xi),
\label{ctfR:a}\\
&\wh{b^{p}_{\ell,j}}:=\wh{a_{j-1}}(\xi)\wh{b^{p}_\ell}(2^{j-1}\xi)= \wh{a}(\xi)\wh{a}(2\xi)\cdots\wh{a}(2^{j-2}\xi)\wh{b^{p}_\ell}(2^{j-1}\xi),
\label{ctfR:bp}\\
&\wh{b^{n}_{\ell,j}}:=\wh{a_{j-1}}(\xi)\wh{b^{n}_\ell}(2^{j-1}\xi)= \wh{a}(\xi)\wh{a}(2\xi)\cdots\wh{a}(2^{j-2}\xi)\wh{b^{n}_\ell}(2^{j-1}\xi).
\label{ctfR:bn}
\end{align}
Note that $a_1=a, b^{p}_{\ell,1}=b_{\ell}^{p}$ and $b^{n}_{\ell,1}=b^{n}_\ell$.
We also define
\[
a_{j;k}:=2^{j/2}a_j(\cdot-2^j k),\quad
b^{p}_{\ell,j;k}:=2^{(j+1)/2} b^{p}_{\ell,j}(\cdot-2^{j+1} k),
\quad b^{n}_{\ell,j;k}:=2^{(j+1)/2} b^{n}_{\ell,j}(\cdot-2^{j+1} k)
\]
for $\ell=1,2$, $j\in \N$, and $k\in \Z$.
Then its associated one-dimensional $J$-level discrete affine system is given by
\begin{align*}
\das_J(\{a \sampled 2; b^{p}_1\sampled 4, b^{p}_2\sampled 4, b^{n}_1\sampled 4, b^{n}_2\sampled 4\})=
\{a_{J;k} \setsp k\in \Z\}
\cup
\{
b^{p}_{\ell,j;k}, b^{n}_{\ell,j;k} \setsp k\in \Z,\ell=1,2, j=1,\ldots, J\}.
\end{align*}

A detailed construction of $\ctf_6^{\reduced}$ is given in the following result by
defining the filters $a$ and $b^{p}_1, b^{p}_2, b^{n}_1, b^{n}_2$ as in \eqref{ctf} with $s=2$ and $a^p, a^n$ as in \eqref{apn}.

\begin{theorem}\label{thm:ctfR6}
Let $0<c_0<c_1<c_2<\pi$ and $\gep_0,\gep_1,\gep_2,\gep_3$ be positive real numbers.
The filters $a, a^p, b^{p}_1, b^{p}_2$ are constructed by defining their $2\pi$-periodic Fourier
series on the basic interval $[-\pi,\pi)$ as follows:
\be \label{ctfR}
\wh{a}:=\chi_{[-c_1,c_1];\gep_1,\gep_1}, \quad \wh{a^p}:=\chi_{[0,c_1];\gep_0,\gep_1}
\quad \mbox{and} \quad \wh{b^{p}_1}:=\chi_{[c_1,c_2];\gep_1,\gep_2},\quad
\wh{b^{p}_2}:=\chi_{[c_2,\pi];\gep_2,\gep_3}.
\ee
Define
\be \label{filter:pn}
a^n:=\ol{a^p},\quad b^{n}_1:=\ol{b^{p}_1}, \quad b^{n}_2:=\ol{b^{p}_2}.
\ee
If
\be \label{ctfR:condition}
\gep_0+\gep_1\le c_1 \le \tfrac{\pi}{2}-\gep_0-\gep_1,\quad
\tfrac{\pi}{2}+\gep_2+\gep_3\le c_2\le \pi-\gep_2-\gep_3,
\quad \gep_1+\gep_2\le c_2-c_1\le \tfrac{\pi}{2}-\gep_1-\gep_2,
\ee
then both $\{a \sampled 2; b^{p}_1\sampled 4, b^{p}_2\sampled 4, b^{n}_1\sampled 4, b^{n}_2\sampled 4\}$ and $\{a^p \sampled 4, a^n \sampled 4; b^{p}_1\sampled 4, b^{p}_2\sampled 4, b^{n}_1\sampled 4, b^{n}_2\sampled 4\}$ are tight framelet filter banks.
If both \eqref{ctfR:condition} and the following additional conditions are satisfied:
\be \label{ctfR:fs}
\tfrac{1}{2}c_2+\tfrac{1}{2} \gep_2+c_1+\gep_1\le \pi \quad\mbox{and}\quad
c_1+\gep_1+\tfrac{1}{2}\gep_3\le \tfrac{\pi}{2},
\ee
then all the high-pass filters $b^{p}_{1,j;k},b^{p}_{2,j;k}, b^{n}_{1,j;k} b^{n}_{2,j;k},k\in \Z$ at all scale levels $j\ge 2$ in the
$J$-level discrete affine system
$\das_J(\{a \sampled 2; b^{p}_1\sampled 4, b^{p}_2\sampled 4, b^{n}_1\sampled 4, b^{n}_2\sampled 4\})$
have the ideal frequency separation property for every $J\ge 2$, more precisely,
\be \label{ctfR:ifs}
\wh{b^{p}_{\ell,j}}(\xi)=0, \quad \forall\, \xi\in [-\pi,0] \quad \mbox{and}\quad
\wh{b^{n}_{\ell,j}}(\xi)=0,\quad \forall\, \xi\in [0,\pi]\qquad \mbox{for all}\; j\ge 2 \quad \mbox{and}\quad \ell=1,2,
\ee
where $\wh{b^{p}_{\ell,j}}$ and $\wh{b^{n}_{\ell,j}}$ are defined in \eqref{ctfR:bp} and \eqref{ctfR:bn}, respectively.
\end{theorem}

\begin{proof} By Theorem~\ref{thm:dft:pr},
$\{a \sampled 2; b^{p}_1\sampled 4, b^{p}_2\sampled 4, b^{n}_1\sampled 4, b^{n}_2\sampled 4\}$
is a tight framelet filter bank if and only if
\begin{align}
&|\wh{a}(\xi)|^2+|\wh{b^{p}_1}(\xi)|^2+|\wh{b^{p}_2}(\xi)|^2
+|\wh{b^{n}_1}(\xi)|^2+|\wh{b^{n}_2}(\xi)|^2
=1,\label{partition}\\
&
\wh{a}(\xi)\ol{\wh{a}(\xi+\pi)}+ \sum_{\ell=1}^2 \Big(\wh{b^{p}_\ell}(\xi)\ol{\wh{b^{p}_\ell}(\xi+\pi)}+
\wh{b^{n}_\ell}(\xi)\ol{\wh{b^{n}_\ell}(\xi+\pi)}\Big)=0,\label{cond0:1}\\
&
\sum_{\ell=1}^2 \Big( \wh{b^{p}_\ell}(\xi)\ol{\wh{b^{p}_\ell}(\xi+\tfrac{\pi}{2})}+
\wh{b^{n}_\ell}(\xi)\ol{\wh{b^{n}_\ell}(\xi+\tfrac{\pi}{2})}\Big)=0,\label{cond0:2}\\
&
\sum_{\ell=1}^2\Big(\wh{b^{p}_\ell}(\xi)\ol{\wh{b^{p}_\ell}(\xi+\tfrac{3\pi}{2})}
+\wh{b^{n}_\ell}(\xi)\ol{\wh{b^{n}_\ell}(\xi+\tfrac{3\pi}{2})}\Big)=0.\label{cond0:3}
\end{align}
By the definition of the bump function, it is easy to check that the identity in \eqref{partition} holds. By our assumption in \eqref{ctfR:condition}, we see that for all $\xi\in \R$,
\be \label{a:nonoverlapping}
\wh{a}(\xi)\wh{a}(\xi+\pi)=0,\quad \wh{a^p}(\xi)\wh{a^p}(\xi+\tfrac{\gamma\pi}{2})=0,\quad
\wh{a^n}(\xi)\wh{a^n}(\xi+\tfrac{\gamma\pi}{2})=0, \qquad \forall\; \gamma=1,2,3
\ee
and
\be \label{b:nonoverlapping}
\wh{u}(\xi)\wh{u}(\xi+\tfrac{\gamma \pi}{2})=0, \qquad \forall\, \gamma=1,2,3,\; u\in \{b^{p}_1, b^{p}_2, b^{n}_1, b^{n}_2\}.
\ee
Therefore, all the three identities in \eqref{cond0:1}--\eqref{cond0:3} trivially hold. Thus, $\{a \sampled 2; b^{p}_1\sampled 4, b^{p}_2\sampled 4, b^{n}_1\sampled 4, b^{n}_2\sampled 4\}$
is a tight framelet filter bank.

By Theorem~\ref{thm:dft:pr}, $\{a^p \sampled 4, a^n \sampled 4; b^{p}_1\sampled 4, b^{p}_2\sampled 4, b^{n}_1\sampled 4, b^{n}_2\sampled 4\}$ is a tight framelet filter bank if and only if
\be \label{partition2}
|\wh{a^p}(\xi)|^2+|\wh{a^n}(\xi)|^2+|\wh{b^{p}_1}(\xi)|^2
+|\wh{b^{p}_2}(\xi)|^2
+|\wh{b^{n}_1}(\xi)|^2+|\wh{b^{n}_2}(\xi)|^2
=1
\ee
and for all $\gamma=1,2,3$,
\be \label{cond0}
\wh{a^{p}}(\xi)\ol{\wh{a^{p}}(\xi+\tfrac{\gamma\pi}{2})}+
\wh{a^{n}}(\xi)\ol{\wh{a^{n}}(\xi+\tfrac{\gamma\pi}{2})}+
\sum_{\ell=1}^2\Big(\wh{b_{\ell}^{p}}(\xi)\ol{\wh{b_{\ell}^{p}}(\xi+\tfrac{\gamma\pi}{2})}
+
\wh{b_{\ell}^{n}}(\xi)\ol{\wh{b_{\ell}^{n}}(\xi+\tfrac{\gamma\pi}{2})}
\Big)=0.
\ee
By the definition of the bump function, it is easy to check that the identity in \eqref{partition2} holds. It also follows directly from \eqref{a:nonoverlapping} and \eqref{b:nonoverlapping} that \eqref{cond0} trivially holds.
Hence, $\{a^p \sampled 4, a^n \sampled 4; b^{p}_1\sampled 4, b^{p}_2\sampled 4, b^{n}_1\sampled 4, b^{n}_2\sampled 4\}$ is a tight framelet filter bank.

Using \eqref{ctfR:condition} and \eqref{ctfR:fs}, by calculation we can
directly check that the ideal frequency separation property in \eqref{ctfR:ifs} holds.
\end{proof}

We now discuss the tensor product tight framelet filter bank $\tpctf_6^{\reduced}$ derived from the one-dimensional tight framelet filter banks in Theorem~\ref{thm:ctfR6}.
Define $\tpctf\mbox{-HP}^{\reduced}_{6}$ to be the set consisting of total
$6^d-2^d$ complex-valued high-pass filters as follows:
\[
\tpctf\mbox{-HP}^{\reduced}_{6}:=\left(\otimes^d \{a^p, a^n, b^{p}_1, b^{p}_2, b^{n}_1, b^{n}_2\}\right)\bs \left(\otimes^d \{a^p, a^n\}\right).
\]
Then the directional tensor product complex tight framelet filter bank $\tpctf_6^{\reduced}$ (called reduced $\tpctf_6$ or $\tpctf_6$ down $4$) for dimension $d$ is defined to be
\be\label{tpctf6R}
\tpctf_6^{\reduced}:=\{\otimes^d a \sampled 2I_d; u \sampled 4I_d\; \mbox{with}\; u\in \tpctf\mbox{-HP}^{\reduced}_{6}\}.
\ee
Note that the low-pass filter $\otimes^d a$ is real-valued and
due to the relations in \eqref{filter:pn}, we see that
$\ol{u}\in \tpctf\mbox{-HP}_{6}^{\reduced}$ for any $u\in \tpctf\mbox{-HP}_{6}^{\reduced}$.
Therefore, we can always rewrite the tight framelet filter bank $\tpctf_6^{\reduced}$ as
\[
\tpctf_6^{\reduced}=\{\otimes^d a\sampled 2I_d;
u\sampled 4I_d, \ol{u}\sampled 4I_d\; \mbox{with}\; u\in \tpctf\mbox{-CHP}_{6}^{\reduced}\},
\]
where $\tpctf\mbox{-CHP}_{6}^{\reduced}$ is a subset of $\tpctf\mbox{-HP}_{6}^{\reduced}$ and has exactly $\tfrac{6^d-2^d}{2}$ filters.
Consequently, the complex-valued tight framelet filter bank $\tpctf_6^{\reduced}$ is essentially equivalent to the following real-valued tight framelet filter bank:
\be \label{ctfR:real}
\{\otimes^d a \sampled 2I_d; \sqrt{2}\re(u) \sampled 4I_d, \sqrt{2}\im(u) \sampled 4I_d \; \mbox{with}\; u\in \tpctf\mbox{-CHP}_{6}^{\reduced}\}.
\ee
Therefore, we essentially have only total $(6^d-2^d)/2$ number of complex-valued high-pass filters in $\tpctf\mbox{-HP}_{6}^{\reduced}$.  Thus, the number of real coefficients (by counting a complex number as two real numbers) produced by all the complex-valued filters in $\tpctf_6^{\reduced}$ is the same as those produced by
the real-valued tight framelet filter bank in \eqref{ctfR:real}.
That is, up to a multiplicative constant $\sqrt{2}$, $\tpctf\mbox{-HP}_{6}^{\reduced}$ produces exactly the same set of real coefficients (by identifying a complex number with two real numbers: its real and imaginary parts) as the $6^d-2^d$ real-valued filters in \eqref{ctfR:real} do.
Note that the sampling matrix is $4I_d$ for all high-pass filters from $\otimes^d\{a^p, a^n, b^{p}_1, b^{p}_2, b^{n}_1, b^{n}_2\}$, while we only perform sampling by $2I_d$ for the low-pass filter $\otimes^d a$. Consequently, regardless of the decomposition level, the redundancy rate of the fast framelet transform employing $\tpctf_6^{\reduced}$ for dimension $d$ is no more than
\[
\frac{6^d-2^d}{4^d} \sum_{j=0}^\infty \frac{1}{2^{dj}}=\frac{3^d-1}{2^d-1}.
\]
For example, the redundancy rates of $\tpctf_6^{\reduced}$ are $2, 2\mathord{\frac{2}{3}}, 3\mathord{\frac{5}{7}}, 5\mathord{\frac{1}{3}}$ and $7\mathord{\frac{25}{31}}$ for $d=1,\ldots, 5$, respectively. See Table~\ref{tab:redundancy} for more details on the redundancy rates of $\tpctf_6^{\reduced}$. Note that the redundancy rate of the original $\tpctf_6$ is $2^d$ times that of $\tpctf_6^{\reduced}$ for dimension $d$.

\section{Numerical Experiments on Image and Video Processing}
\label{sec:experiments}

In this section, we shall test the performance of our constructed directional tensor product complex tight framelet $\tpctf_6^{\reduced}$ with low redundancy in Section~\ref{sec:tpctf} and compare it with many other frame-based methods for image and video processing such as the denoising and inpainting problems.


For the directional tensor product complex tight framelet $\tpctf_6^{\reduced}$ with low redundancy that will be used in this section for image and video processing, the parameters in Theorem~\ref{thm:ctfR6} are set to be
\be \label{ctf6R:parameters}
c_1 = \tfrac{\pi}{2}-0.425, \quad
c_2 = 2.0,\quad
\varepsilon_0 = 0.125, \quad
\varepsilon_1 = 0.3, \quad
\varepsilon_2 = 0.35, \quad
\varepsilon_3 = 0.0778.
\ee
Note that the above parameters satisfy the conditions in both \eqref{ctfR:condition} and \eqref{ctfR:fs}. The parameters for other $\tpctf_m$ are set to be the same as those in the paper \cite{HanZhao:siims:2014}. For the convenience of the reader, we explicitly list these parameters here: For $\tpctf_3$, we set
\[
c_1=\tfrac{33}{32}, \quad c_2=\pi, \quad \gep_1=\tfrac{69}{128}, \quad \gep_2=\tfrac{51}{512}.
\]
For $\tpctf_6$, we set
\[
c_1=\tfrac{119}{128}, \quad c_2=\tfrac{\pi}{2}+\tfrac{119}{256},\quad c_3=\pi,\quad \gep_0=\tfrac{35}{128},\quad \gep_1=\tfrac{81}{128}, \quad \gep_2=\tfrac{115}{256},\quad \gep_3=\tfrac{115}{256}.
\]
To have some ideas about the filters in $\ctf_6^{\reduced}$, see Figure~\ref{fig:ctf6R:1D} for the frequency response of the filters in $\ctf_6^{\reduced}$.
For the directionality of $\tpctf_6^{\reduced}$ in dimension two, see Figure~\ref{fig:ctf6R:2D} for some elements of $\das_J(\tpctf_6^{\reduced})$ in dimension two with $J=5$.

\begin{figure}[th]
\centering
\includegraphics[width=9cm,height=4cm]{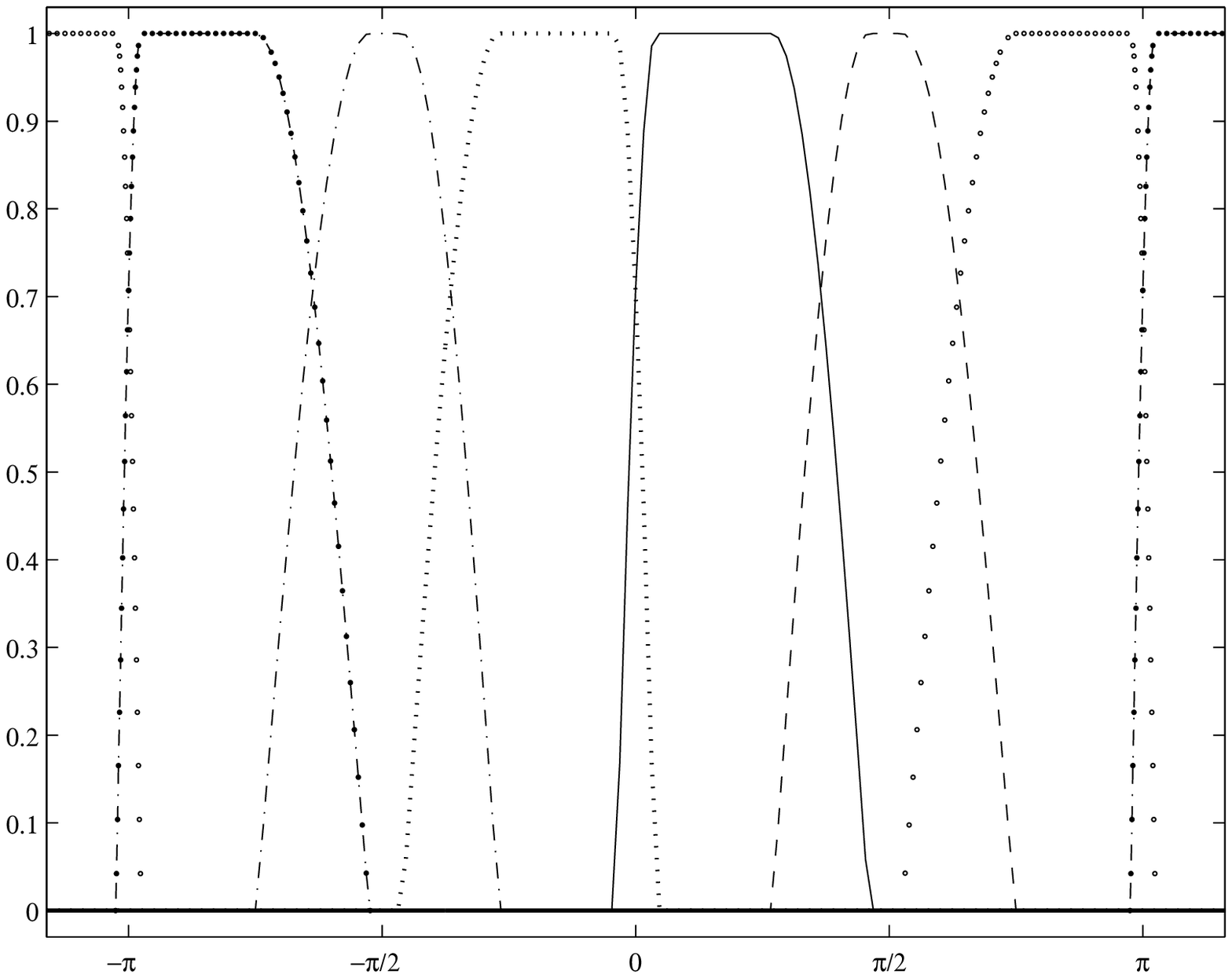}
\caption{The one-dimensional tight framelet filter bank $\ctf_6^{\reduced}=\{a^p \sampled 4,a^n \sampled 4; b^{p}_1 \sampled 4,b^{p}_2\sampled 4,b^{n}_1\sampled 4,b^{n}_2\sampled 4\}$ in Theorem~\ref{thm:ctfR6} with parameters in \eqref{ctf6R:parameters}.  Solid line for $\wh{a^p}$, dotted line for $\wh{a^n}$, dashed line for $\wh{b^{p}_1}$, dash-dotted line for $\wh{b^{n}_1}$, circled line for $\wh{b^{p}_2}$, and
circle-dotted line for $\wh{b^{n}_2}$.}\label{fig:ctf6R:1D}
\end{figure}

\begin{figure}[th]
\centering
\includegraphics[width=16cm,height=8cm]{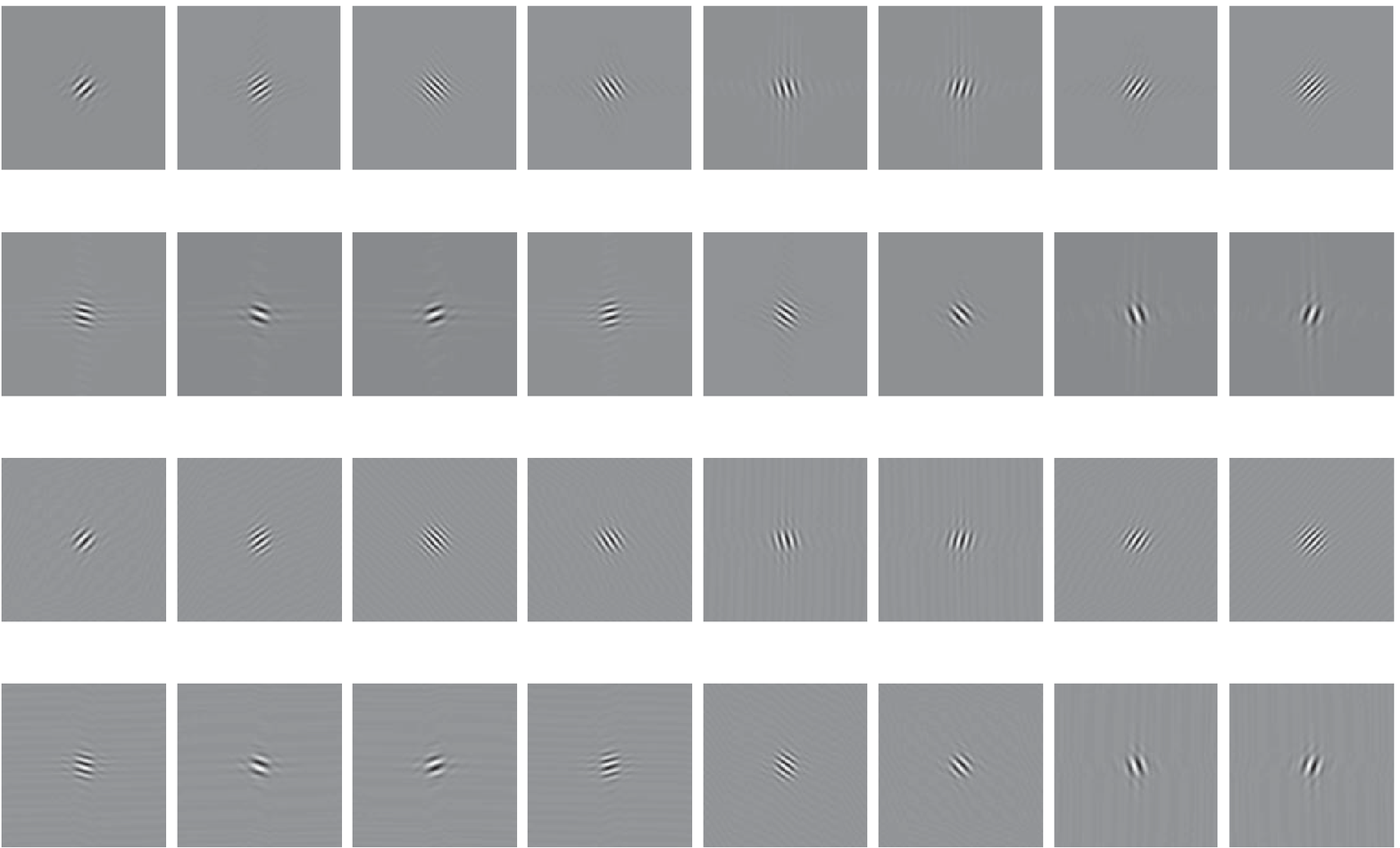}
\caption{The first two rows show the real part and the last two rows show the imaginary part of the 2D high-pass filters at the level $4$ in $\das_5(\tpctf_6^{\reduced})$ for dimension two. Among these $16$ graphs for the first two rows or the last two rows, the directions along $\pm 45^\circ$ are repeated once. Hence, there are a total of $14$ directions in the 2D discrete affine system $\das_5(\tpctf_6^{\reduced})$.}\label{fig:ctf6R:2D}
\end{figure}

As usual, the performance is measured by the peak signal-to-noise ratio (PSNR) which is defined to be
\begin{equation}\label{def:PSNR}
\mathrm{PSNR}(u,\mathring u) = 10\log_{10}\frac{255^2}{\mathrm{MSE}(u-\mathring u)},
\end{equation}
where $u$ is an original/true image supported on $[1,N]^2$, $\mathring u$ is a reconstructed data, and $\mathrm{MSE}(u-\mathring{u})=\frac{1}{N^2}\sum_{j=1}^{N} \sum_{k=1}^{N} |u(j,k)-\mathring{u}(j,k)|^2$ is the mean squared error.

\subsection{Image denoising and image inpainting}
We first compare the performance of $\tpctf_6^{\reduced}$ for image denoising. We compare the performance of $\tpctf_6^{\reduced}$ with two groups of different approaches. The first group uses tensor-product approach including
$\tpctf_3$ (which has the same redundancy rate $2\frac23$ as that of $\tpctf_6^{\reduced}$) and $\tpctf_6$ (which has the same directionality as $\tpctf_6^{\reduced}$ but has a higher redundancy rate $10\frac{2}{3}$), as well as the dual tree complex wavelet transform ($\dtcwt$) in \cite{SBK} (which has the redundancy rate $4$). The second group employs non-tensor-product approach including curvelets, shearlets, and smooth affine shear tight frames. Curvelets in \cite{FDCT} and compactly supported shearlets in \cite{LimIEEE,LimIEEE2} can be downloaded from the corresponding authors' websites. We download each of their packages and run their denoising codes for test images. Smooth affine shear tight frames (ASTF) are developed by two of the authors of this paper in \cite{HanZhuang:acha:2014}.

The CurveLab package at \texttt{http://www.curvelab.org} has two subpackages: one uses un-equispace FFT and the other uses frequency wrapping. Here we use the frequency wrapping package; detailed information on CurveLab package can be found in \cite{FDCT}. The  performance of these two subpackages are very close to each other (less than $0.2$dB differences) and here we choose the one with the frequency wrapping for comparison. The total number of scales is 5. At the finest scale level, the CurveLab uses an isotropic wavelet transform to avoid checkerboard effect.  At the scale level 4, 32 (angular) directions are used. At the scale levels 3 and 2, 16 (angular) directions are used. At the coarsest scale level, 8 (angular) directions are used. The redundancy rate of the CurveLab wrapping package is about 2.8.

 The ShearLab package at \texttt{http://www.shearlab.org} also has many subpackages for different implementations. Here we choose two subpackages using compactly supported shearlets. One is DST as described in \cite{LimIEEE} and the other is DNST as described in \cite{LimIEEE2}.  The DNST in \cite{LimIEEE2} has the best performance so far in the ShearLab package. For DST, the total number of scales is 5. Ten shear directions are used across all scale levels. The redundancy rate of the DST is 40. For DNST, the total number of scales is 4. Sixteen shear directions are used for the finest scale levels 4 and 3; while 8 shear directions are used for the other two scale levels. All filters are implemented in an undecimated fashion. The redundancy rate of DNST is 49.

For the smooth affine shear tight frames (ASTF) in \cite{HanZhuang:acha:2014},  we use total 16 shear directions for the finest scale level.  For the next three scales, we use 8 shear directions, and for the coarsest scale level, we use 4 shear directions. The redundancy rate of this system is about 5.4. See \cite{HanZhuang:acha:2014} for more details.

The decomposition levels for all directional tensor product complex tight framelets $\tpctf_m$ are set to be $J=5$, while the decomposition level for the dual tree complex wavelet transform is set to be $J=6$ (see \cite{SBK,SS:bslocal}). We use symmetric boundary extension for all test images to avoid the boundary effect with the boundary extension size for all test images being $16$ pixels. The strategy for processing frame coefficients for all tensor-product transforms is the bivariate shrinkage proposed in \cite{SS:bslocal} with window size $7\times 7$ and constant $\sqrt{3}$. Let $\sigma$ denote the standard deviation of the i.i.d. Gaussian noise.
More precisely, a frame coefficient $c$ is processed by the bivariate shrinkage function $\eta_\lambda^{bs}$ as follows:
\be \label{bs}
\eta^{bs}_\lambda(c)=\eta_{\lambda_c}^{soft}(c)
=\begin{cases}
c-\lambda_c \tfrac{c}{|c|}, &|c|> \lambda_c,\\
0, &\text{otherwise},
\end{cases}
\qquad \mbox{with}\quad
\lambda_c:=\frac{\sqrt{3}\sigma_n^2}{\sigma_c \sqrt{1+|c_p/c|^2}},
\ee
where $\sigma_n:=\sigma \|b\|_2$ with $b$ being the high-pass filter inducing the frame coefficient $c$,
the frame coefficient $c_p$ is the parent coefficient of $c$ in the immediate higher scale, and
\[
\sigma_c:=\begin{cases} \sqrt{\breve{\sigma}_c^2-\sigma_n^2}, &\breve{\sigma}_c> \sigma_n,\\
0, &\text{otherwise}
\end{cases}
\qquad \mbox{with}\qquad
\breve{\sigma}_c^2:=\frac{1}{\#N_c}\sum_{j\in N_c} |c_j|^2,
\]
where $\#N_c$ is the cardinality of the set $N_c$ which is the $[-3,3]^2$ window centering around the frame coefficient $c$ at the band induced by the filter $b$.

See Figure~\ref{fig:testimages} for the four $512\times 512$ grayscale test images: \emph{Barbara}, \emph{Lena}, \emph{Fingerprint}, and \emph{Boat}. The comparison results of performance are reported in
Table~\ref{tab:imagedenoising} for image denoising under independent identically distributed
Gaussian noise with noise standard deviation $\sigma=5,10,25,40,50,80,100$.

\begin{figure}[th]
\centering
\subfigure[Barbara]{\includegraphics[width=2.5cm]{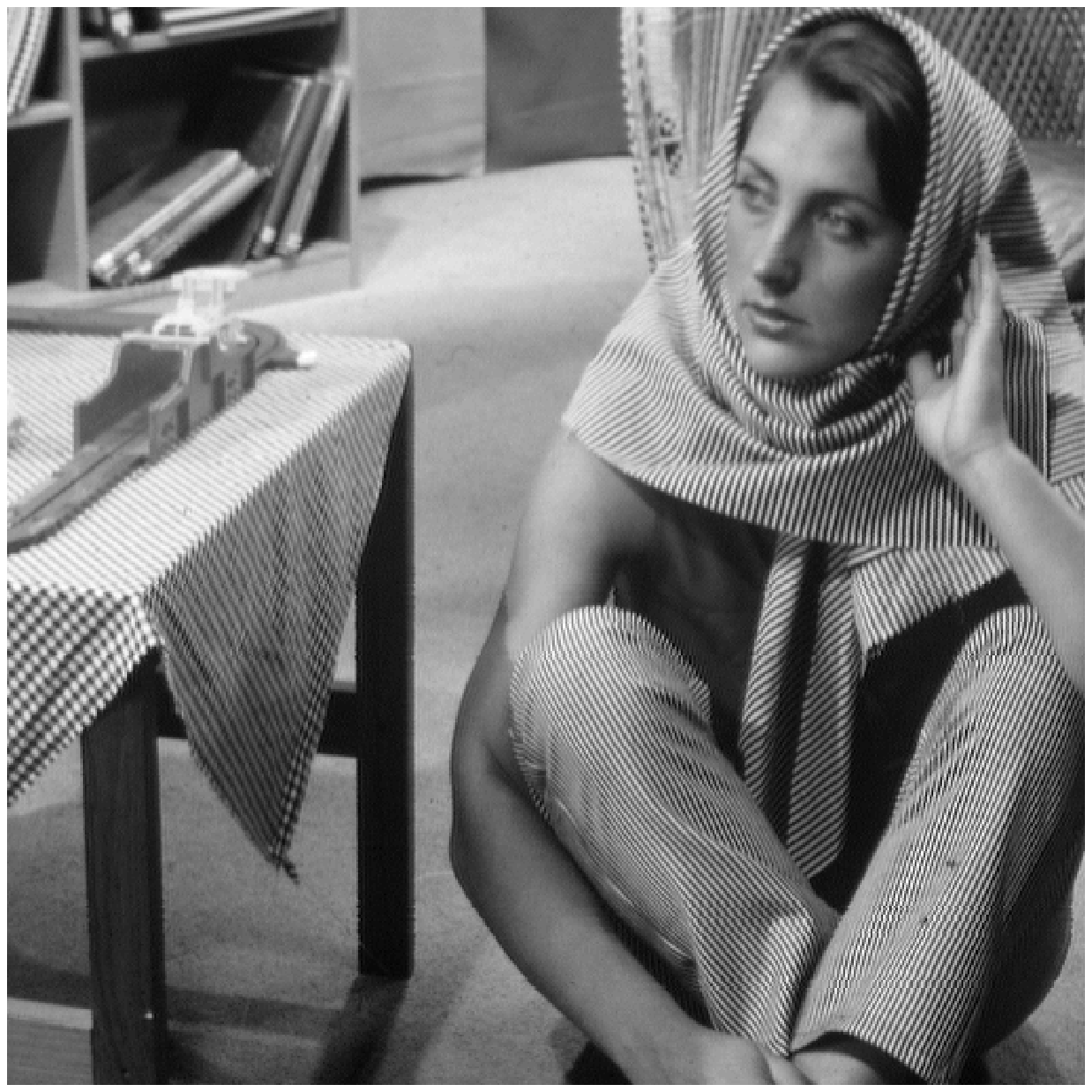}}
\subfigure[Lena]{\includegraphics[width=2.5cm]{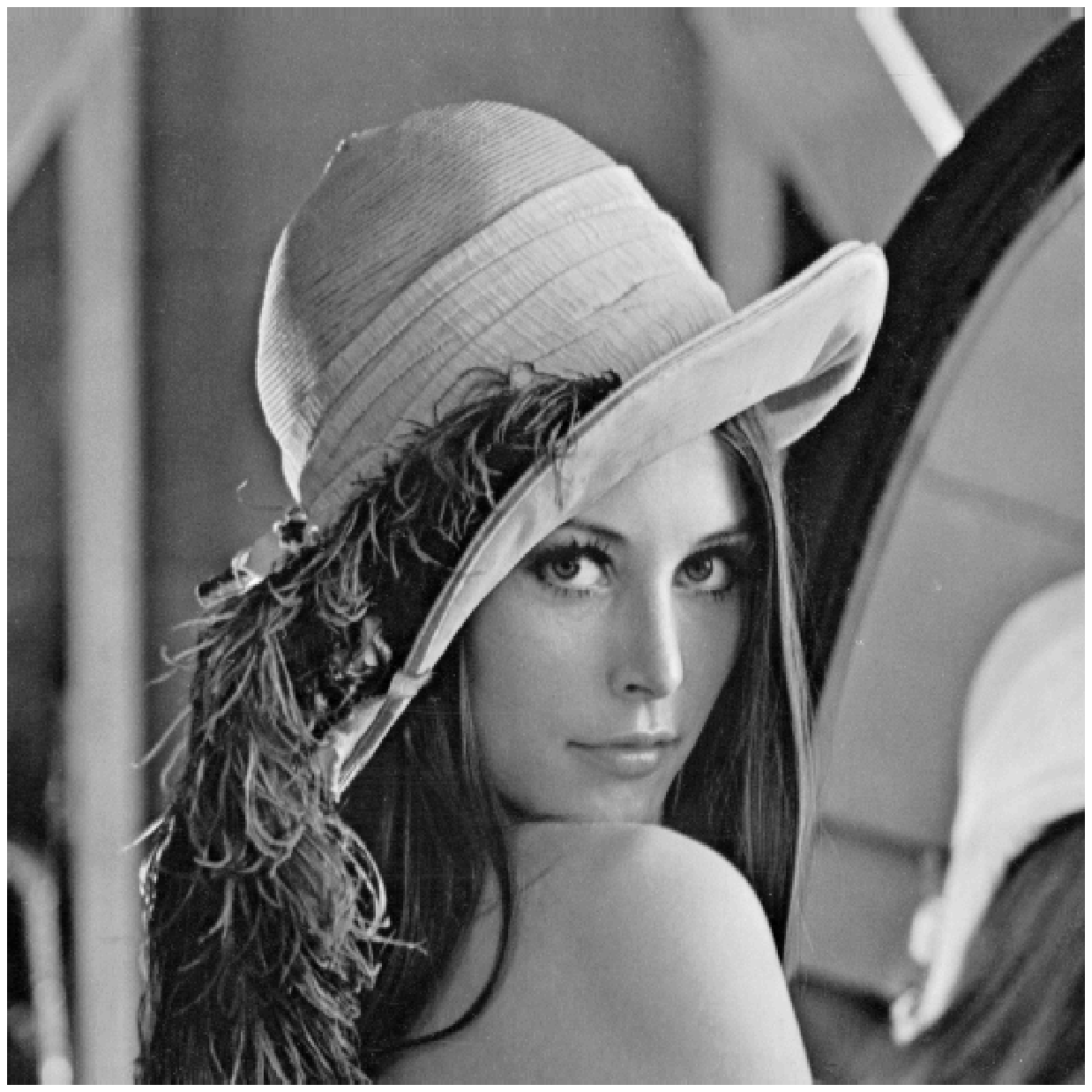}}
\subfigure[Fingerprint]{\includegraphics[width=2.5cm]{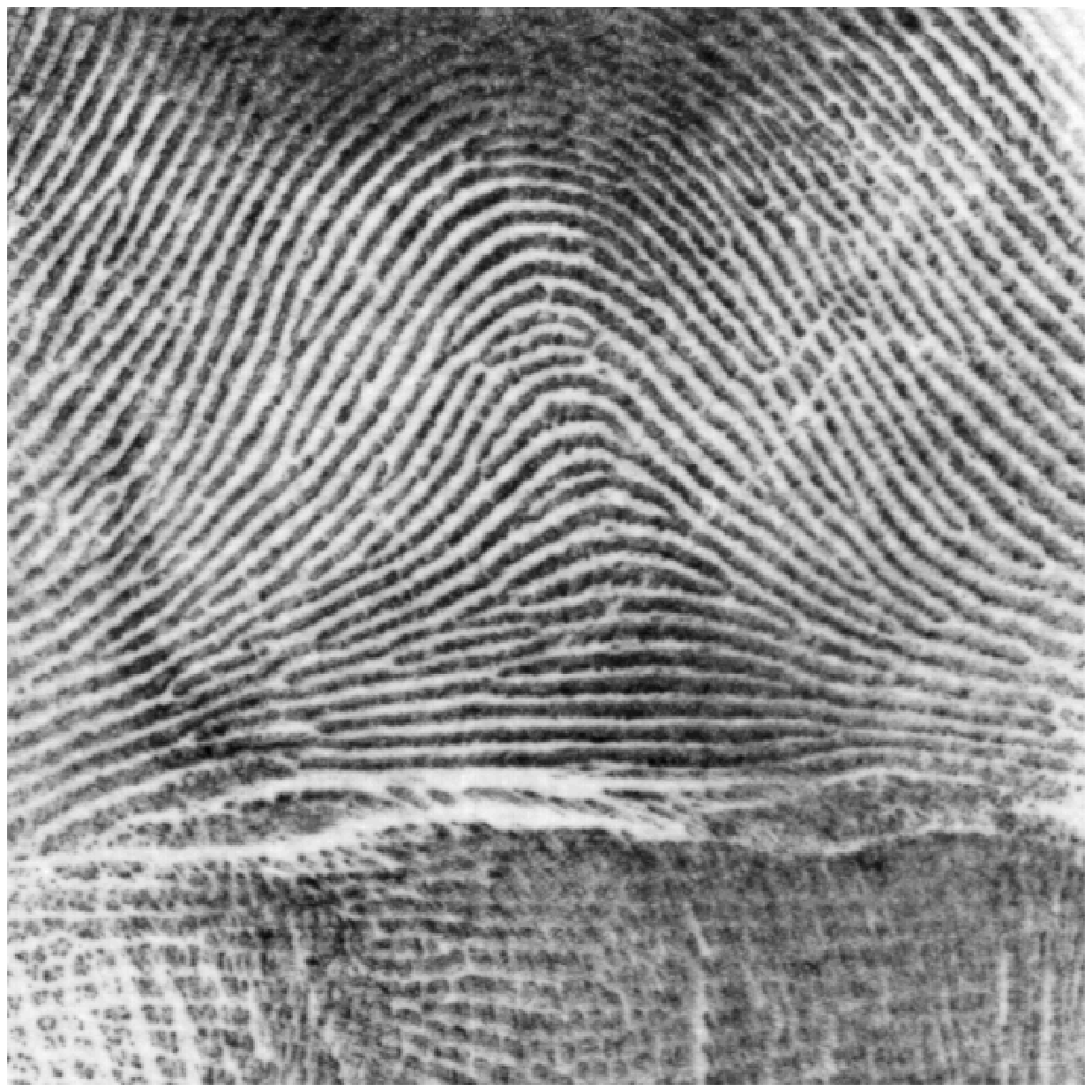}}
\subfigure[Boat]{\includegraphics[width=2.5cm]{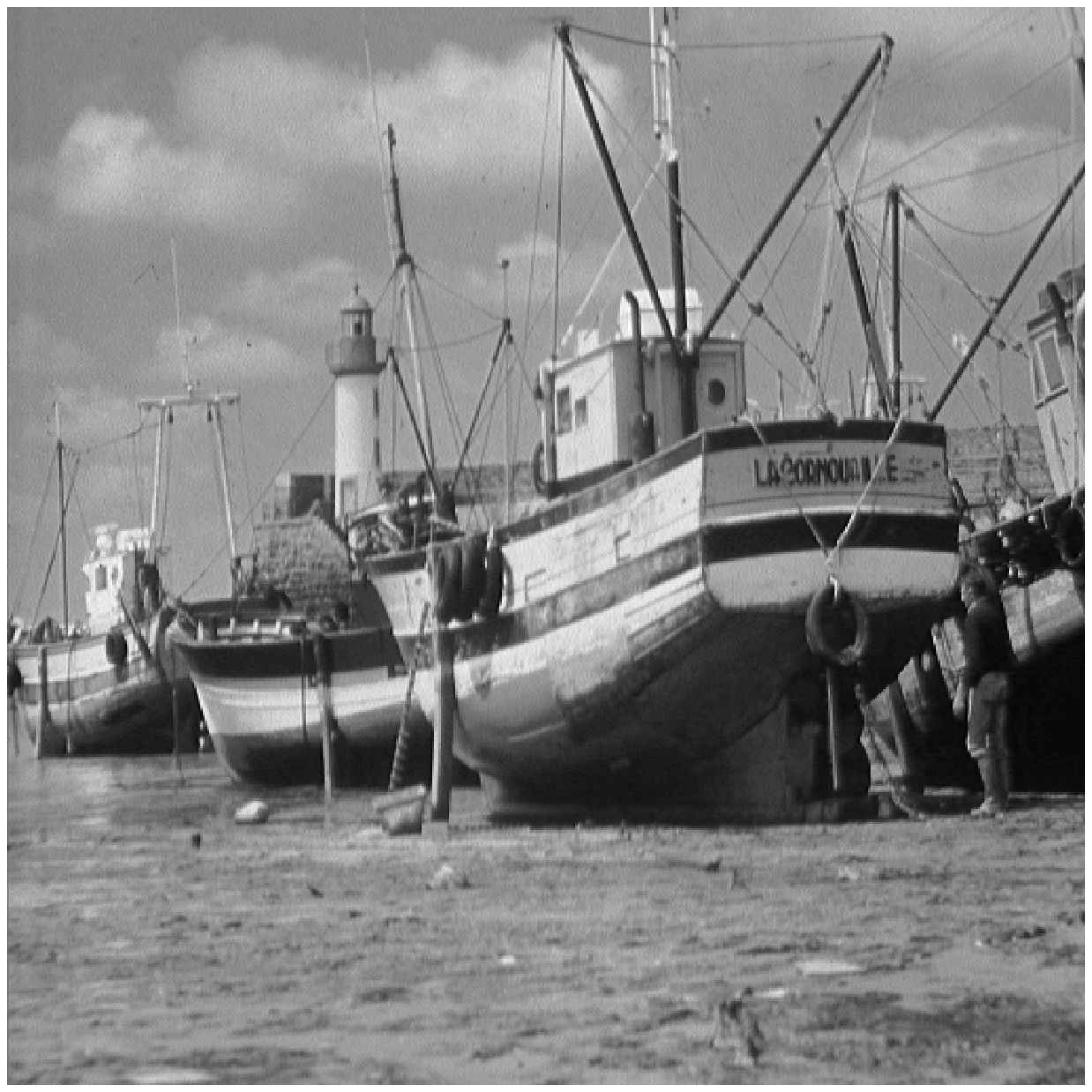}}\\
\subfigure[Mobile]{\includegraphics[width=2.5cm]{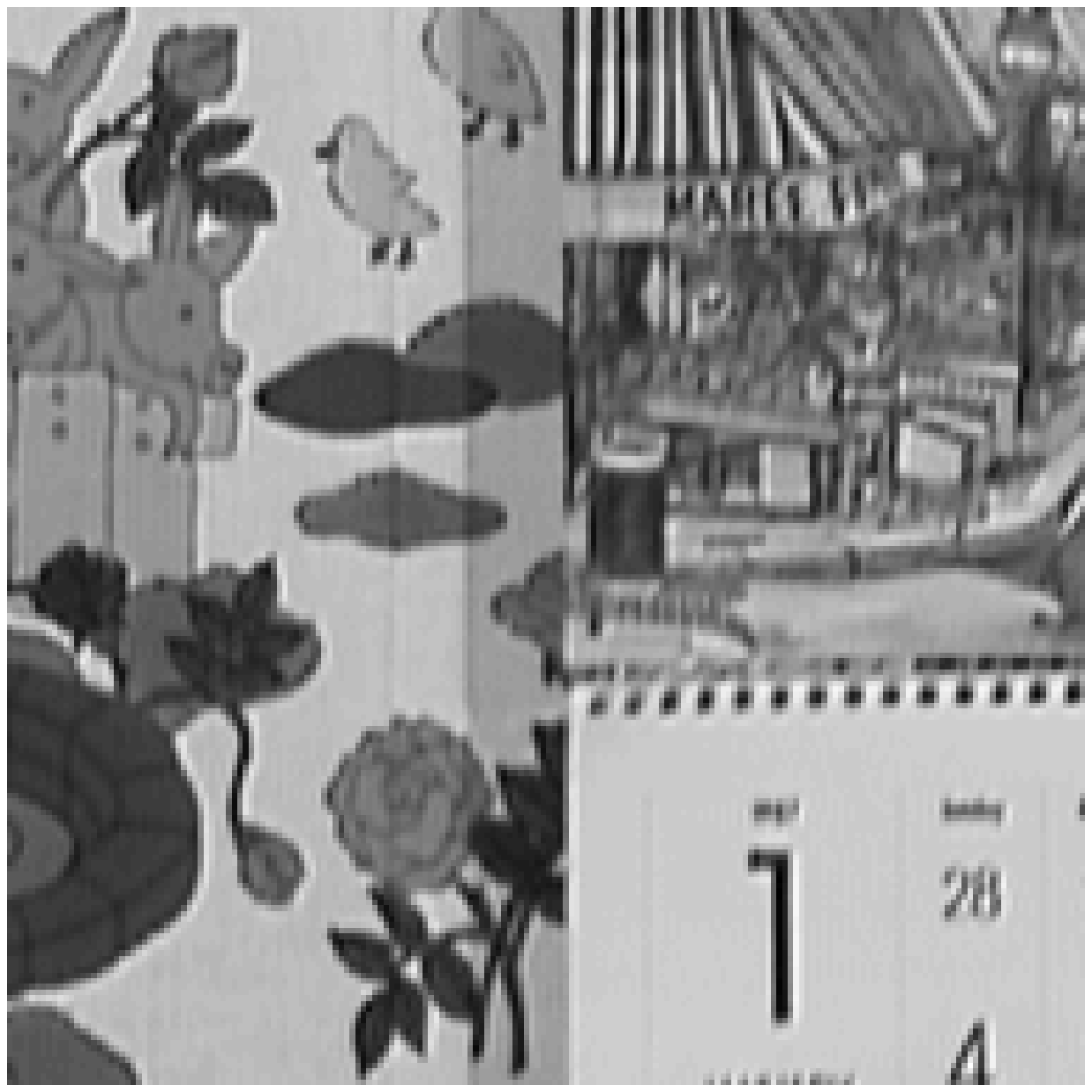}}
\subfigure[Coastguard]{\includegraphics[width=2.5cm]{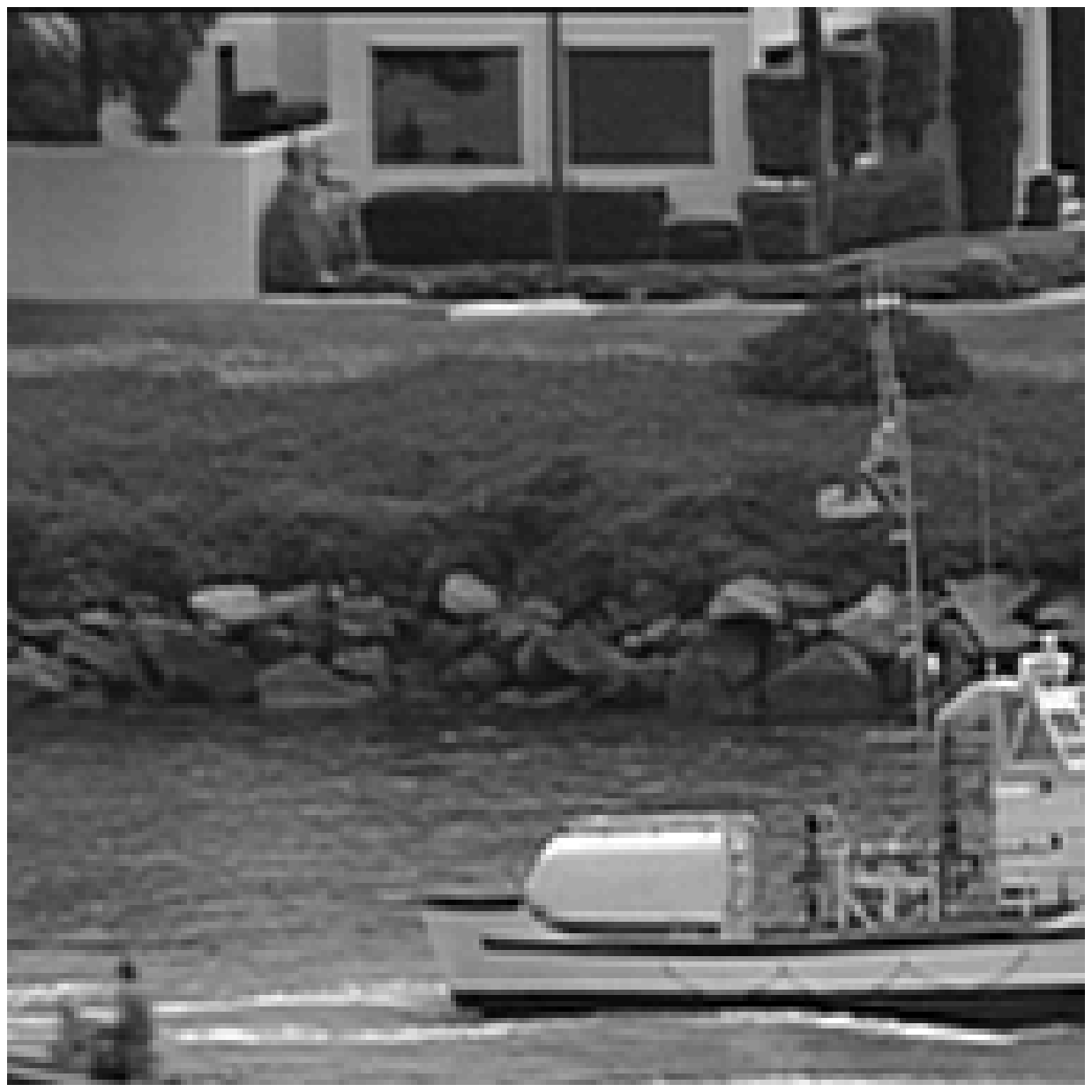}}
\subfigure[Text~1]{\includegraphics[width=2.5cm]{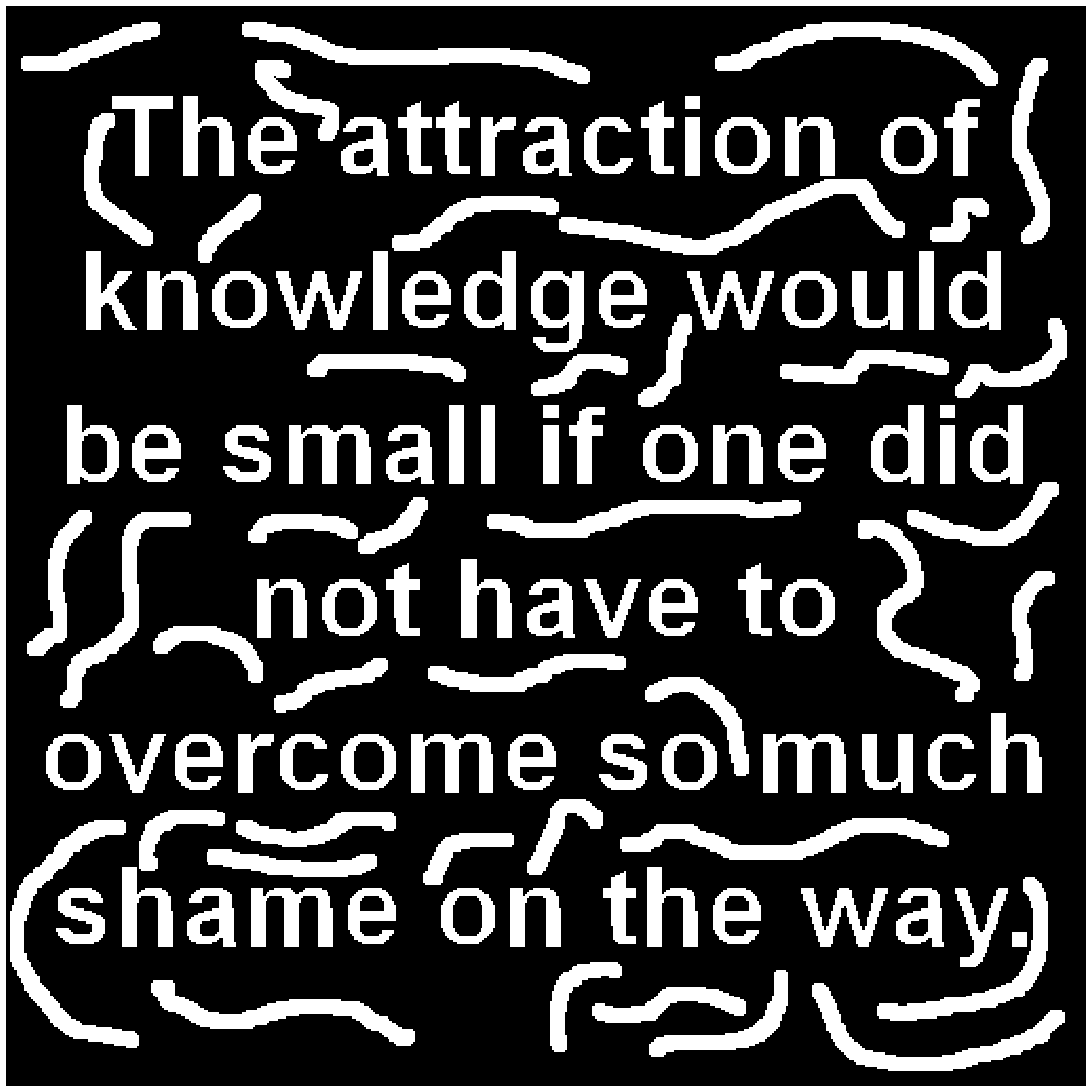}}
\subfigure[Text~2]{\includegraphics[width=2.5cm]{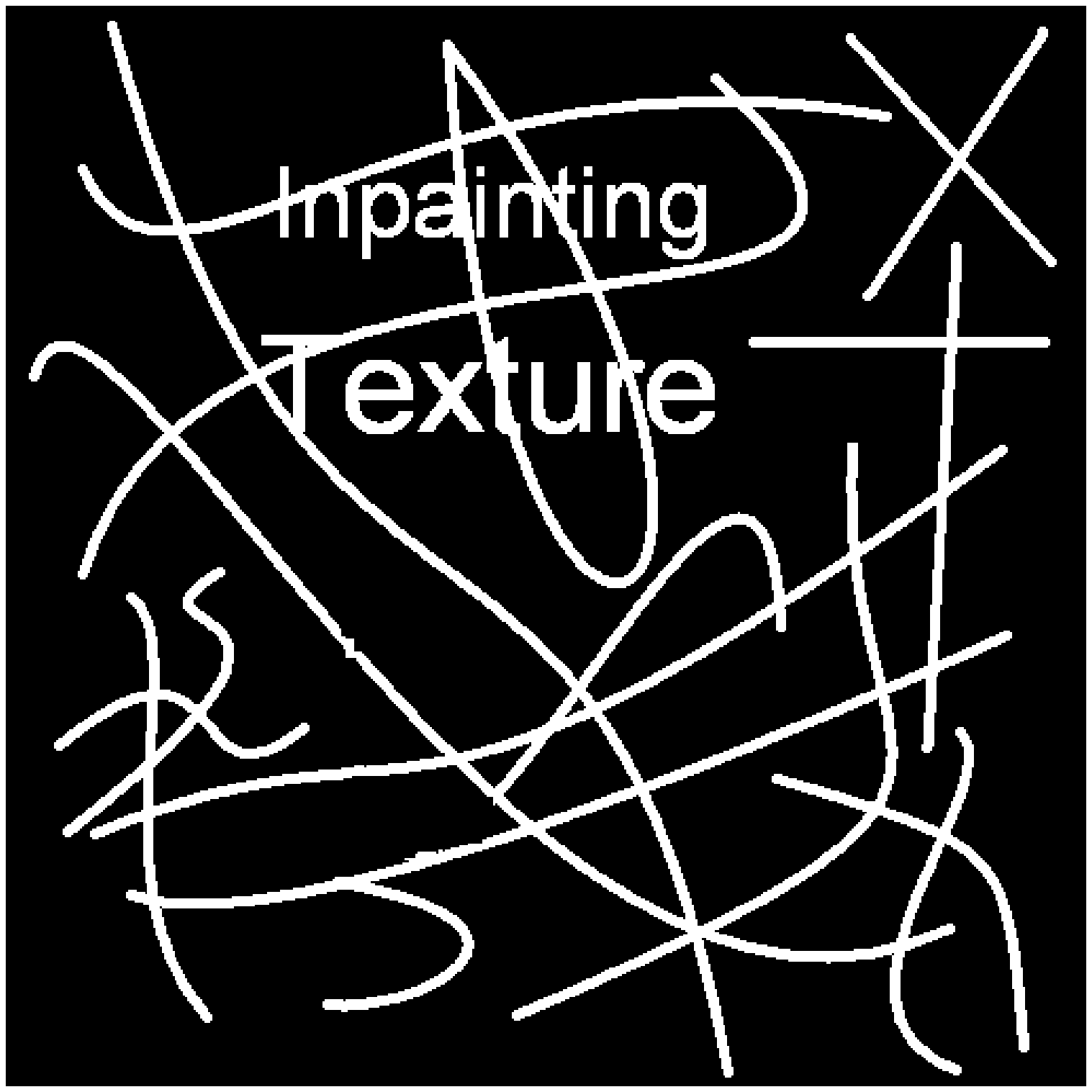}}
\caption{(a)-(d) are the four $512\times 512$ grayscale test images: \emph{Barbara}, \emph{Lena}, \emph{Fingerprint}, and \emph{Boat}. (e)-(f) are the first frame of the $192\times 192\times 192$ videos: \emph{Mobile} and \emph{Coastguard}. (e) and (f) are inpainting masks of size $512\times 512$.
}\label{fig:testimages}
\end{figure}

\begin{table}[htbp]
  \begin{small}
  \begin{center}
    \begin{tabular}{|c||c||c|c|c||c|c|c|c|}
\hline
    &\multicolumn{8}{|c|}{$512\times 512$ Barbara} \\
    \hline
    $\sigma$ & $\tpctf_6^{\reduced}$  & $\tpctf_6$& $\tpctf_3$& $\dtcwt$& CurveLab  & DST& DNST& ASTF\\
    \hline
    5     & 37.63 & 37.84(-0.21) & 37.16(0.47) & 37.37(0.26) & 33.83(3.80) & 37.76(-0.13) & 37.17(0.46) & 37.40(0.23)\\
    10    & 33.97 & 34.18(-0.21) & 33.19(0.78) & 33.54(0.43) & 29.17(4.80) & 33.94(0.03) & 33.62(0.35) & 33.74(0.23)\\
    25    & 29.28 & 29.35(-0.07) & 28.04(1.24) & 28.81(0.47) & 24.83(4.45) & 28.90(0.38) & 28.93(0.35) & 29.29(-0.01)\\
    40    & 26.85 & 26.86(-0.01) & 25.53(1.32) & 26.45(0.40) & 23.87(2.98) & 26.36(0.49) & 26.48(0.37) & 27.08(-0.23)\\
    50    & 25.73 & 25.71(0.02)  & 24.48(1.25) & 25.36(0.37) & 23.38(2.35) & 25.22(0.51) & 25.31(0.42) & 26.05(-0.32)\\
    80    & 23.51 & 23.53(-0.02) & 22.82(0.69) & 23.27(0.24) & 22.22(1.29) & 23.11(0.40) & 22.96(0.55) & 23.97(-0.46)\\
    100   & 22.58 & 22.64(-0.06) & 22.25(0.33) & 22.42(0.16) & 21.61(0.97) & 22.23(0.35) & 22.06(0.52) & 23.02(-0.44)\\ \hline
    &\multicolumn{8}{|c|}{$512\times 512$ Lena}\\
    \hline
    5      & 38.16 & 38.37(-0.21) & 37.98(0.18) & 38.25(-0.09) & 35.77(2.39) & 38.22(-0.06	) & 38.01(0.15) & 38.19(-0.03) \\
    10     & 35.22 & 35.48(-0.26) & 34.93(0.29) & 35.19(0.03)  & 33.37(1.85) & 35.19(0.03) & 35.35(-0.13) & 35.18(0.04) \\
    25     & 31.20 & 31.60(-0.40) & 31.17(0.03) & 31.29(-0.09) & 30.07(1.13) & 31.09(0.11) & 31.51(-0.31) & 31.40(-0.20) \\
    40     & 29.10 & 29.52(-0.42) & 29.24(-0.14) & 29.22(-0.12)& 28.15(0.95) & 28.92(0.18) & 29.32(-0.22) & 29.40(-0.30) \\
    50     & 28.11 & 28.54(-0.43) & 28.34(-0.23) & 28.22(-0.11)& 27.19(0.92) & 27.89(0.22) & 28.21(-0.10) & 28.46(-0.35) \\
    80     & 26.11 & 26.47(-0.36) & 26.42(-0.31) & 26.15(-0.04)& 25.16(0.95) & 25.71(0.40) & 25.78(0.33) & 26.44(-0.34) \\
    100    & 25.21 & 25.52(-0.31) & 25.52(-0.31) & 25.20(0.01) & 24.22(0.99) & 24.67(0.54) & 24.58(0.63) & 25.48(-0.27) \\ \hline
    &\multicolumn{8}{|c|}{$512\times 512$ Fingerprint} \\
    \hline
    5     & 36.29 & 36.27(0.02) & 35.29(1.00) & 35.82(0.47) & 33.35(2.94) & 36.02(0.27) & 35.28(1.01) & 35.20(1.09) \\
    10    & 32.23 & 32.10(0.13) & 30.97(1.26) & 31.74(0.49) & 30.61(1.62) & 31.95(0.28) & 31.76(0.47) & 30.97(1.26) \\
    25    & 27.27 & 26.98(0.29) & 26.56(0.71) & 27.26(0.01) & 26.03(1.24) & 27.04(0.23) & 27.10(0.17) & 26.95(0.32) \\
    40    & 25.02 & 24.68(0.34) & 24.75(0.27) & 24.98(0.04) & 23.92(1.10) & 24.79(0.23) & 24.82(0.20) & 25.01(0.01) \\
    50    & 24.01 & 23.67(0.34) & 23.84(0.17) & 23.95(0.06) & 23.00(1.01) & 23.77(0.24) & 23.78(0.23) & 24.07(-0.06) \\
    80    & 21.99 & 21.66(0.33) & 21.73(0.26) & 21.91(0.08) & 21.18(0.81) & 21.65(0.34) & 21.63(0.36) & 22.11(-0.12) \\
    100   & 21.09 & 20.75(0.34) & 20.69(0.40) & 21.01(0.08) & 20.37(0.72) & 20.63(0.46) & 20.56(0.53) & 21.22(-0.13) \\ \hline
    &\multicolumn{8}{|c|}{$512\times 512$ Boat}\\
    \hline
    5      & 36.74 & 36.92(-0.18) & 36.45(0.29) & 36.73(0.01)   & 33.59(3.15) & 36.51(0.23	) & 36.04(0.70) & 36.66(0.08) \\
    10     & 33.10 & 33.41(-0.31) & 32.97(0.13) & 33.19(-0.09)  & 30.60(2.50) & 33.07(0.03) & 33.15(-0.05) & 33.07(0.03) \\
    25     & 28.81 & 29.26(-0.45) & 28.98(-0.17) & 29.03(-0.22) & 27.51(1.30) & 28.75(0.06	) & 29.23(-0.42) & 29.10(-0.29) \\
    40     & 26.72 & 27.19(-0.47) & 26.98(-0.26) & 26.99(-0.27) & 25.96(0.76) & 26.71(0.01	) & 27.20(-0.48) & 27.14(-0.42) \\
    50     & 25.79 & 26.25(-0.46) & 26.07(-0.28) & 26.06(-0.27) & 25.18(0.61) & 25.78(0.01	) & 26.23(-0.44) & 26.23(-0.44) \\
    80     & 24.05 & 24.41(-0.36) & 24.29(-0.24) & 24.22(-0.17) & 23.55(0.50) & 23.90(0.15	) & 24.17(-0.12) & 24.41(-0.36) \\
    100    & 23.27 & 23.58(-0.31) & 23.50(-0.23) & 23.39(-0.12) & 22.79(0.48) & 23.05(0.22	) & 23.17(0.10) & 23.57(-0.30) \\
    \hline
\end{tabular}%
\medskip
\caption{Comparison results, in terms of PSNR values, of several image denoising methods using our proposed directional tensor product complex tight framelet $\tpctf_6^{\reduced}$ with the redundancy rate $2\tfrac{2}{3}$, tensor product complex tight framelets $\tpctf_6$ with the redundancy rate $10\tfrac{2}{3}$ (having the same directionality as $\tpctf_6^{\reduced}$), $\tpctf_3$ with redundancy rate $2\tfrac{2}{3}$ (having the same redundancy rate as $\tpctf_6^{\reduced}$), dual tree complex wavelet transform with the redundancy rate $4$ in \cite{SBK},
CurveLab (Wrap) with redundancy rate 2.8 in \cite{FDCT}, DST with redundancy rate 40 in \cite{LimIEEE}, DNST with redundancy rate 49 in \cite{LimIEEE2}, and ASTF with redundancy rate 5.8 in \cite{HanZhuang:acha:2014}. The $\tpctf_6^{\reduced}, \tpctf_6,\tpctf_3,\dtcwt$ are separable transforms using tensor product tight frames while the CurveLab, DST, DNST, ASTF are nonseparable transforms using 2D non-tensor-product (tight) frames.
The values in parentheses are the PSNR gain/loss of $\tpctf_6^{\reduced}$ over the compared transform: positive numbers in parentheses mean that $\tpctf_6^{\reduced}$ performs better than the corresponding transform, while negative numbers in parentheses mean that $\tpctf_6^{\reduced}$ performs worse than the corresponding transform.}
\label{tab:imagedenoising}%
\end{center}
\end{small}
\end{table}%

For texture-rich test images such as \emph{Barbara} and \emph{Fingerprint}, we can see from Table~\ref{tab:imagedenoising} that $\tpctf_6^{\reduced}$ outperforms  $\tpctf_3$,  $\dtcwt$, CurveLab, DST and DNST. It can have up to 1.32dB PSNR value improvement over $\tpctf_3$ for \emph{Barbara} at $\sigma = 40$ and about 0.5dB improvement over $\dtcwt$ for \emph{Fingerprint} at $\sigma = 10$. In comparison with $\tpctf_6$,  $\tpctf_6^{\reduced}$ outperforms $\tpctf_6$ for the test image \emph{Fingerprint} for all $\sigma$ noise levels but has slightly worse performance than $\tpctf_6$ for the test image of \emph{Barbara}. CurveLab (Wrap) also has low redundancy rate, yet its performance is not as good as others for all the test images. DST and DNST have high redundancy rates almost 20 times of that of $\tpctf_6^{\reduced}$. However, for such  images of \emph{Barbara} and \emph{Fingerprint}, the performance of DST and DNST is not as good as our $\tpctf_6^\reduced$. With redundancy rate about 2 times of $\tpctf_6^\reduced$, the performance of ASTF is better than $\tpctf_6^\reduced$ only when the noise level is high $\sigma>40$.

It can be seen from the test images of \emph{Lena} and \emph{Boat} in Figure~\ref{fig:testimages} that most of their edges are concentrating along the horizontal, the vertical, or the two diagonal directions. For such test images, when $\sigma$ is small ($\sigma<40$), the performance of $\tpctf_6^{\reduced}$ is almost the same as  $\tpctf_3$ and $\dtcwt$. Only when $\sigma$ is high ($\sigma\ge 40$), $\tpctf_6^{\reduced}$ performs not as well as $\tpctf_3$ and $\dtcwt$, but generally within less than 0.3dB loss of performance. For comparison among $\tpctf_6^{\reduced}$ and $\tpctf_6$, DNST, ASTF, we see at most 0.48dB loss of performance of $\tpctf_6^{\reduced}$ for both \emph{Lena} and \emph{Boat}. $\tpctf_6^\reduced$ outperform DST and CurveLab  for the test images of \emph{Lena} and \emph{Boat}.

$\tpctf_6$ has recently been used in \cite{SHB:2014} for the image inpainting problem with impressive
performance over many other inpainting algorithms.
Here we simply use the same inpainting algorithm as developed in \cite{SHB:2014} but with $\tpctf_6$
being replaced by $\tpctf_6^{\reduced}$. As similar to most frame-based inpainting algorithms in the literature,
the inpainting algorithm in \cite{SHB:2014} uses iterative thresholding algorithm with gradually decreasing threshold values.
For a detailed description of the inpainting algorithm using $\tpctf_6$,
see \cite{SHB:2014}. For image inpainting without noise,
here we only compare the performance of our $\tpctf_6^{\reduced}$ with three state-of-the-art inpainting algorithms: (1) \cite{SHB:2014} using $\tpctf_6$
with redundancy rate $10\tfrac{2}{3}$.
(2) \cite{LiShenSuter:ieee:2013} using a tight frame built
from the undecimated DCT-Haar wavelet filter which is derived from the discrete cosine transform (DCT) with a block size $7\times 7$ and has the redundancy rate $49$.
(3) DNST in \cite{LimIEEE2} using undecimated compactly supported nonseparable shearlet frames which has the redundancy rate $49$ with $16, 16, 8, 8$ high-pass filters and one low-pass filter. See \cite{CCS:framelet:inpainting,CCS:ipi:2010,LiShenSuter:ieee:2013,LimIEEE2,SHB:2014} for image inpainting and comparison results with other frame-based image inpainting algorithms. The inpainting algorithms in \cite{LiShenSuter:ieee:2013,LimIEEE2,SHB:2014} have been generously provided to us by their own authors.
The numerical results on image inpainting without noise are presented in Table~\ref{tab:image:inpainting}.

\begin{table}[htbp]
\begin{small}
\begin{center}
\begin{tabular}{|c||c||c|c|c||c||c|c|c|}
\hline
&\multicolumn{4}{|c||}{$512\times 512$ Barbara} &\multicolumn{4}{c|}{$512\times 512$ Lena}\\
\hline
& \!$\tpctf_6^{\reduced}$\!  & \! \cite{SHB:2014}($\tpctf_6$)\! &\! \cite{LiShenSuter:ieee:2013}\! & \!\cite{LimIEEE2}(DNST)\!
& \!$\tpctf_6^{\reduced}$\!  & \!\cite{SHB:2014}($\tpctf_6$)\! & \!\cite{LiShenSuter:ieee:2013}\!
& \! \cite{LimIEEE2}(DNST)\! \\ \hline
\mbox{Text~1} & 36.68 & 36.59(0.09) & 35.03(1.65) & 35.17(1.51) & 37.71 & 38.02(-0.31) & 36.73(0.98) & 38.17(-0.46)\\
\mbox{Text~2} & 32.99 & 32.68(0.31) & 31.51(1.48) & 32.45(0.54) & 33.92 & 34.31(-0.39) & 32.10(1.82) & 34.10(-0.18) \\
50\% & 35.75 & 35.73(0.02) & 33.85(1.90) & 34.13(1.62) & 37.68 & 38.00(-0.32) & 37.65(0.03) & 36.49(1.19) \\
80\%  & 28.55 & 28.16(0.39) & 26.39(2.16) & 28.22(0.33) & 31.99 & 32.33(-0.34) & 30.55(1.44) & 31.64(0.35)\\
   \hline
&\multicolumn{4}{|c||}{$512\times 512$ Fingerprint} &\multicolumn{4}{c|}{$512\times 512$ Boat}\\ \hline
\mbox{Text~1} & 31.87 & 31.35(0.52) & 30.44(1.43) & 31.05(0.82) & 34.57 & 34.96(-0.18) & 34.62(0.29) & 34.66(0.01) \\
\mbox{Text~2} & 28.36 & 27.78(0.58) & 26.11(2.25) & 27.17(1.19) & 30.39 & 30.80(-0.31) & 30.35(0.13) & 30.65(-0.09) \\
50\% & 34.19 & 34.12(0.07) & 33.26(0.93) & 31.18(3.01) & 34.00 & 34.42(-0.45) & 34.08(-0.17) & 33.07(-0.22) \\
80\% & 26.77 & 26.00(0.77) & 25.72(1.05) & 25.38(1.39) & 28.03 & 28.58(-0.55) & 27.89(0.14) & 28.01(0.02) \\ \hline
\end{tabular}%
\medskip
\caption{Performance in terms of PSNR values of several image inpainting algorithms without noise.
The first two rows for Text~1 and Text~2 are for the inpainting masks Text~1 and Text~2
in Figure~\ref{fig:testimages}. The last two rows are for $50\%$ or $80\%$ randomly
missing pixels. \cite{SHB:2014} uses $\tpctf_6$ with the redundancy rate $10\tfrac{2}{3}$. $\tpctf_6^{\reduced}$ uses
the same inpainting algorithm as in \cite{SHB:2014} but with $\tpctf_6$
being replaced by $\tpctf_6^{\reduced}$ which has the redundancy rate $2\tfrac{2}{3}$.
\cite{LiShenSuter:ieee:2013} uses a tight frame built
from the undecimated DCT-Haar wavelet filter with redundancy rate $49$.
\cite{LimIEEE2} uses undecimated compactly supported nonseparable shearlet frames with redundancy rate $49$.
}
\label{tab:image:inpainting}%
\end{center}
\end{small}
\end{table}%

We now look at the image inpainting problem with i.i.d. Gaussian noise.
The image inpainting algorithm proposed in \cite{SHB:2014} using $\tpctf_6$ not only performs well for image inpainting without noise but also is stable and works well for the image inpainting problem with noise. On the other hand, most available image inpainting algorithms (e.g., \cite{CCS:framelet:inpainting,CCS:ipi:2010,LiShenSuter:ieee:2013,LimIEEE2} and references therein) are not stable and barely work well for image inpainting with noise.
Similar to image denoising, for image inpainting with noise, though the noise level $\sigma$ can be effectively estimated, we assume for simplicity that the noise level $\sigma$ is known in advance. Such an assumption is commonly adopted in the literature.
As pointed out in \cite{SHB:2014}, there are no parameters to be tuned in the image inpainting algorithm proposed in \cite{SHB:2014} and the image inpainting without noise is simply a special case by taking $\sigma=0$.

As we did for image inpainting without noise, for image inpainting with noise,
we not only test the performance of the inpainting algorithm in \cite{SHB:2014} using $\tpctf_6$ and our modified algorithm using $\tpctf_6^{\reduced}$ but also run the inpainting algorithms in
\cite{LiShenSuter:ieee:2013} and \cite{LimIEEE2}. However, for the noise levels $\sigma=10,\ldots,50$, the inpainting algorithms in \cite{LiShenSuter:ieee:2013} and \cite{LimIEEE2}
often have significantly lower performance than \cite{SHB:2014}. For example, for the test image of Barbara with inpainting mask Text~1 and with the noise level $\sigma=50$, the performance of PNSR values are
$24.91$, $24.91$, $14.66$, $14.48$ for $\tpctf_6^{\reduced}$, \cite{SHB:2014} using $\tpctf_6$, \cite{LiShenSuter:ieee:2013}, \cite{LimIEEE2}, respectively. This indicates that the inpainting algorithms in \cite{LiShenSuter:ieee:2013} and \cite{LimIEEE2} are mainly designed for the image inpainting problem without noise. As a consequence, for image inpainting with noise, we do not report the comparison results using inpainting algorithms in \cite{LiShenSuter:ieee:2013} and \cite{LimIEEE2}. Instead, we only report experimental comparison results using $\tpctf_6^{\reduced}$ with the inpainting algorithm in \cite{SHB:2014} using $\tpctf_6$ in Table~\ref{tab:imageinpaintingnoise}.

\begin{table}[htbp]
  \begin{small}
  \begin{center}
    \begin{tabular}{|c||c|c||c|c||c|c||c|c|} \hline
    &\multicolumn{8}{c|}{$512\times512$ Barbara}\\ \hline
    &\multicolumn{2}{|c||}{Text~1} &\multicolumn{2}{c||}{Text~2} &\multicolumn{2}{c||}{50\% missing} &\multicolumn{2}{c|}{80\% missing}\\ \hline
    $\sigma$ & \!$\tpctf_6^{\reduced}$\!& \!\cite{SHB:2014}($\tpctf_6$)\! & \! $\tpctf_6^{\reduced}$\! &\!\cite{SHB:2014}($\tpctf_6$)\! & \!$\tpctf_6^{\reduced}$\!  & \!\cite{SHB:2014}($\tpctf_6$)\! &\!$\tpctf_6^{\reduced}$\! & \! \cite{SHB:2014}($\tpctf_6$)\!\\
    \hline
    10    & 31.76 & 31.81(-0.05) & 29.99 & 29.85(0.14) & 30.94 & 31.11(-0.17) & 26.56 & 26.70(-0.14) \\
    20    & 29.00 & 28.99(0.01) & 27.76 & 27.71(0.05) & 27.94 & 28.00(-0.06) & 24.48 & 24.70(-0.22) \\
    30    & 27.21 & 27.18(0.03) & 26.24 & 26.24(0.00) & 25.95 & 25.95(0.00) & 23.18 & 23.34(-0.16) \\
    40    & 25.91 & 25.88(0.03) & 25.10 & 25.14(-0.04) & 24.58 & 24.56(0.02) & 22.14 & 22.45(-0.31) \\
    50    & 24.91 & 24.91(0.00) & 24.18 & 24.30(-0.12) & 23.59 & 23.60(-0.01) & 21.42 & 21.90(-0.48) \\
   \hline
    &\multicolumn{8}{c|}{$512\times512$ Lena}\\ \hline
    10    & 33.08 & 33.42(-0.34) & 31.32 & 31.80(-0.48) & 32.86 & 33.40(-0.54) & 29.34 & 30.25(-0.91) \\
    20    & 30.83 & 31.26(-0.43) & 29.58 & 30.10(-0.42) & 30.19 & 30.84(-0.65) & 27.33 & 28.36(-1.03) \\
    30    & 29.32 & 29.81(-0.49) & 28.34 & 28.89(-0.55) & 28.52 & 29.18(-0.66) & 25.94 & 26.95(-1.01) \\
    40    & 28.21 & 28.72(-0.51) & 27.37 & 27.97(-0.60) & 27.35 & 27.98(-0.63) & 24.92 & 25.93(-1.01) \\
    50    & 27.33 & 27.85(-0.52) & 26.60 & 27.22(-0.62) & 26.39 & 27.06(-0.67) & 24.11 & 25.15(-1.04) \\
   \hline
    &\multicolumn{8}{c|}{$512\times512$ Fingerprint}\\ \hline
    10    & 28.77 & 28.46(0.31) & 26.67 & 26.24(0.43) & 29.09 & 28.88(0.21) & 24.60 & 24.12(0.48) \\
    20    & 26.46 & 26.20(0.26) & 25.11 & 24.72(0.39) & 26.10 & 25.76(0.34) & 22.93 & 22.49(0.54) \\
    30    & 24.98 & 24.70(0.28) & 23.99 & 23.59(0.40) & 24.43 & 24.07(0.36) & 21.81 & 21.51(0.30) \\
    40    & 23.90 & 23.61(0.29) & 23.10 & 22.72(0.38) & 23.10 & 22.91(0.19) & 20.96 & 20.68(0.28) \\
    50    & 23.05 & 22.76(0.29) & 22.39 & 22.00(0.39) & 22.39 & 22.01(0.38) & 20.26 & 19.96(0.30) \\
    \hline
    &\multicolumn{8}{c|}{$512\times512$ Boat}\\ \hline
    10    & 30.65 & 31.04(-0.39) & 28.40 & 28.80(-0.40) & 30.11 & 30.65(-0.54) & 26.23 & 27.08(-0.85) \\
    20    & 28.40 & 28.84(-0.44) & 26.88 & 27.32(-0.44) & 27.61 & 28.20(-0.59) & 24.76 & 25.56(-0.80) \\
    30    & 26.95 & 27.41(-0.46) & 25.79 & 26.24(-0.45) & 26.07 & 26.66(-0.59) & 23.75 & 24.46(-0.71) \\
    40    & 25.90 & 26.38(-0.48) & 24.98 & 25.43(-0.45) & 25.01 & 25.56(-0.55) & 23.05 & 23.60(-0.55) \\
    50    & 25.11 & 25.57(-0.46) & 24.32 & 24.80(-0.48) & 24.23 & 24.75(-0.52) & 22.41 & 22.96(-0.55) \\
   \hline
\end{tabular}%
\medskip
\caption{Performance in terms of PSNR values for image inpainting with the noise level $\sigma=10,\ldots,50$. $\tpctf_6^{\reduced}$ uses the same inpainting algorithm as in \cite{SHB:2014} by replacing $\tpctf_6$ with $\tpctf_6^{\reduced}$.
}
\label{tab:imageinpaintingnoise}%
\end{center}
\end{small}
\end{table}%

The experimental results in Tables~\ref{tab:image:inpainting} and~\ref{tab:imageinpaintingnoise}
show that $\tpctf_6^{\reduced}$ performs as well as $\tpctf_6$ in \cite{SHB:2014} for image inpainting with or without noise.
Both $\tpctf_6^{\reduced}$ and $\tpctf_6$ often outperform the state-of-the-art inpainting algorithms in \cite{LiShenSuter:ieee:2013,LimIEEE2}.

\subsection{Video denoising and video inpainting}
For video denoising in three dimensions, the directional tensor product complex tight framelet $\tpctf_6^{\reduced}$ has the redundancy rate $3\tfrac{5}{7}$.
We compare the performance of $\tpctf_6^{\reduced}$ with
the directional tensor product complex tight framelet $\tpctf_6$ (which has the same directionality as $\tpctf_6^{\reduced}$ but has the redundancy rate $29\frac57$), $\tpctf_3$ (which has the same redundancy rate $3\frac{5}{7}$ as $\tpctf_6^{\reduced}$), the 3D dual tree complex wavelet transform ($\dtcwt$, which has the redundancy rate $8$), the 3D nonseparable surfacelets in \cite{LD} (which has the redundancy rate $6.4$), and the 3D nonseparable compactly supported shearlet frames $\mbox{DNST}^{3D}$-$42$ and $\mbox{DNST}_2^{3D}$-$154$ in \cite{LimIEEE2} in ShearLab with $\mbox{DNST}^{3D}$-$42$ and $\mbox{DNST}_2^{3D}$-$154$ having the redundancy rates $42$ and $154$, respectively.
The decomposition level for all tensor product complex tight framelets $\tpctf_m$ is set to be $J=4$ and the boundary extension size for all $\tpctf_m$ is set to be $16$ pixels.
The strategy for processing frame coefficients for all $\tpctf_m$ and $\dtcwt$ is the same bivariate shrinkage as outlined in \eqref{bs} but with window size $3$ instead of $7$.
The constant $\sqrt{3}$ in the bivariate shrinkage function in \eqref{bs} for $\dtcwt$ is still set to be $\sqrt{3}$, but this constant is replaced by $\sqrt{4}$ for $\tpctf_m$ (though there are no significant performance differences if the constant $\sqrt{3}$ is used for $\tpctf_m$). All parameters for 3D surfacelets and the two 3D shearlets $\mbox{DNST}^{3D}$-$42$ and $\mbox{DNST}_2^{3D}$-$154$ are the same as those described in \cite{LimIEEE2,LD}.
The two video sequences \emph{Mobile} and \emph{Coastguard} are used for comparison, which are the same test videos as used in the paper \cite{LimIEEE2} and can be downloaded from the ShearLab 3D package at \texttt{http://www.shearlab.org}.
See Figure~\ref{fig:testimages} for the first frame of these two videos \emph{Mobile} and \emph{Coastguard}.
The comparison results of performance are reported in
Table~\ref{tab:video} under independent identically distributed
Gaussian noise with noise standard deviation $\sigma=10,20,30,40,50,80,100$.

\begin{table}[htbp]
\begin{small}
  \centering
    \begin{tabular}{|c||c||c|c|c|c|c|c|}
    \hline
    & \multicolumn{7}{c|}{$192\times 192\times 192$ Mobile} \\
    \hline
    $\sigma$ & $\tpctf_6^{\reduced}$&
$\tpctf_6$
 & $\tpctf_3$
 &$\dtcwt$
 &\mbox{Surfacelets} &$\mbox{DNST}^{3D}$-$42$  &$\mbox{DNST}^{3D}$-$154$ \\
    \hline
    10    & 35.26 & 35.52(-0.26) & 33.40(1.86) & 34.11(1.15) & 32.79(2.47) & 35.27(-0.01) & 35.91(-0.65) \\
    20    & 31.58 & 31.77(-0.19) & 29.90(1.68) & 30.53(1.05) & 29.95(1.63) & 31.32(0.26) & 32.18(-0.60) \\
    30    & 29.52 & 29.66(-0.14) & 28.03(1.51) & 28.55(0.97) & 28.26(1.26) & 29.00(0.52) & 29.99(-0.47) \\
    40    & 28.10 & 28.20(-0.10) & 26.76(1.34) & 27.17(0.93) & 27.05(1.05) & 27.37(0.73) & 28.42(-0.32) \\
    50    & 27.01 & 27.08(-0.07) & 25.79(1.22) & 26.15(0.86) & 26.11(0.90) & 26.13(0.88) & 27.22(-0.21) \\
    80    & 24.82 & 24.82(0.00) & 23.87(0.95) & 24.03(0.79) & 24.25(0.57) & 23.69(1.13) & 24.75(0.07) \\
    100   & 23.87 & 23.82(0.05) & 23.06(0.81) & 23.06(0.81) & 23.40(0.47) & 22.63(1.24) & 23.62(0.25) \\
    \hline
    & \multicolumn{7}{c|}{$192\times 192\times 192$ Coastguard} \\
\hline
    10    & 33.86 & 34.15(-0.29) & 32.59(1.27) & 33.16(0.70) & 30.86(3.00) & 33.13(0.73) & 33.81(0.05) \\
    20    & 30.26 & 30.62(-0.36) & 29.21(1.05) & 29.66(0.60) & 28.26(2.00) & 29.45(0.81) & 30.28(-0.02) \\
    30    & 28.38 & 28.73(-0.35) & 27.46(0.92) & 27.82(0.56) & 26.87(1.51) & 27.50(0.88) & 28.40(-0.02) \\
    40    & 27.13 & 27.45(-0.32) & 26.28(0.85) & 26.58(0.53) & 25.91(1.21) & 26.17(0.96) & 27.13(-0.00) \\
    50    & 26.18 & 26.48(-0.30) & 25.40(0.78) & 25.66(0.52) & 25.17(1.01) & 25.17(1.01) & 26.17(0.01) \\
    80    & 24.30 & 24.53(-0.23) & 23.67(0.63) & 23.84(0.46) & 23.61(0.69) & 23.17(1.13) & 24.17(0.13) \\
    100   & 23.47 & 23.65(-0.18) & 22.91(0.56) & 22.98(0.49) & 22.87(0.60) & 22.24(1.23) & 23.22(0.25) \\
    \hline
    \end{tabular}%
\medskip
\caption{
Comparison results, in terms of PSNR values, of several video denoising methods using our proposed 3D directional tensor product complex tight framelet $\tpctf_6^{\reduced}$ with the redundancy rate $3\tfrac{5}{7}$, 3D tensor product complex tight framelet $\tpctf_6$ with the redundancy rate $29\tfrac{5}{7}$ (having the same directionality as $\tpctf_6^{\reduced}$), $\tpctf_3$ with the redundancy rate $3\tfrac{5}{7}$ (having the same redundancy rate as $\tpctf_6^{\reduced}$),
the 3D dual tree complex wavelet transform ($\dtcwt$) with the redundancy rate $8$, the 3D nonseparable
surfacelets in \cite{LD} with the redundancy rate $6.4$, and the 3D nonseparable compactly supported shearlet frames $\mbox{DNST}^{3D}$-$42$ and $\mbox{DNST}_2^{3D}$-$154$ with the redundancy rates $42$ and $154$, respectively.
}\label{tab:video}%
\end{small}
\end{table}%

From Table~\ref{tab:video}, we see that the loss of performance of $\tpctf_6^{\reduced}$ is not significant in comparison with $\tpctf_6$ for both \emph{Mobile} and \emph{Coastguard}. $\tpctf_6^{\reduced}$ can outperform $\mbox{DNST}_2^{3D}$-$154$ when the noise level $\sigma$ is high ($\sigma>50$) despite the fact that $\mbox{DNST}_2^{3D}$-$154$ has the highest redundancy rate $154$ which is $41.5$ times the redundancy rate of $\tpctf_6^{\reduced}$. Generally, $\tpctf_6^{\reduced}$ outperforms all other methods (excluding $\tpctf_6$) for any noise level $\sigma$ (except a slightly worse performance at $\sigma=10$ comparing with  $\mbox{DNST}^{3D}$-$42$ for \emph{Mobile}).
Significant improvement can be seen in comparison with the nonseparable 3D surfacelets in \cite{LD} (up to 2.47dB for \emph{Mobile} and 3dB for \emph{Coastguard})
and $\mbox{DNST}^{3D}$-$42$ in \cite{LimIEEE2} (up to 1.24dB for \emph{Mobile} and 1.23dB for \emph{Coastguard}).

For video inpainting, we  use the same inpainting algorithm as developed in \cite{SHB:2014} but with 2D tensor product complex tight framelet $\tpctf_6$ and $\tpctf_6^{\reduced}$ being replaced by 3D tensor product complex tight framelet $\tpctf_6$ and $\tpctf_6^{\reduced}$, respectively. We  compare the performance of $\tpctf_6^{\reduced}$ with  surfacelets in \cite{LD} and 3D nonseparable compactly supported shearlet frames $\mbox{DNST}^{3D}$-$42$ and $\mbox{DNST}^{3D}$-$154$  in ShearLab 3D package.  The numerical results on video inpainting are presented in Table~\ref{tab:video:inpainting}.

\begin{table}[htbp]
\begin{small}
\begin{center}
\begin{tabular}{|c||c||c|c|c|c|}
\hline
&\multicolumn{5}{|c|}{$192\times 192 \times 192$ Mobile (50\% missing)} \\
\hline
$\sigma$& \!$\tpctf_6^{\reduced}$\!  & $\tpctf_6$ &Surfacelets & $\mbox{DNST}^{3D}$-$42$ & $\mbox{DNST}^{3D}$-$154$\\ \hline
0     & 41.15 & 41.74(-0.59)   & 32.09(9.06) & 39.54(1.61)  &  40.71(0.44)\\
10    & 32.65 & 33.09(-0.44)   & 24.70(7.95) & 28.94(3.71)  &  29.20(3.45)  \\
30    & 27.56 & 27.87(-0.31)   & 16.35(11.21) & 20.08(7.48) &  20.35(7.21) \\
   \hline
&\multicolumn{5}{|c|}{$192\times 192 \times 192$ Mobile (80\% missing)} \\ \hline
0     & 28.22 & 28.61(-0.39)   & 22.27(5.95) & 31.09(-2.87) &  33.21(-4.99)  \\
10    & 27.32 & 27.84(-0.52)   & 20.47(6.85) & 27.60(-0.28) &  28.45(-1.13)  \\
30    & 22.89 & 23.53(-0.64)   & 15.81(7.08) & 21.27(1.62)  &  21.60(1.29)\\
\hline
&\multicolumn{5}{|c|}{$192\times 192 \times 192$ Coastguard (50\% missing)} \\
\hline
0     & 37.19 & 37.75(-0.56)   & 28.67(8.52)  & 35.74(1.45)     & 36.69(0.50)\\
10      & 30.88 & 31.48(-0.60)   & 23.61(7.27)  & 28.17(2.71)   & 28.51(2.37)\\
30       & 26.59 & 27.15(-0.56)   & 16.13(10.46)  & 19.92(6.67) & 20.17(6.42)\\
   \hline
&\multicolumn{5}{|c|}{$192\times 192 \times 192$ Coastguard (80\% missing)} \\ \hline
0       & 26.63 & 27.41(-0.78)   & 20.96(5.67)  & 28.56(-1.93)   & 30.02(-3.39)\\
10       & 26.07 & 26.67(-0.60)   & 19.73(6.34)  & 26.18(-0.11)  & 26.92(-0.85)\\
30       & 22.68 & 23.29(-0.61)   & 15.81(6.87)  & 20.87(1.81)   & 21.10(1.58) \\
\hline
\end{tabular}%
\medskip
\caption{
Performance in terms of PSNR values of several video inpainting algorithms.
Gaussian noise with noise levels are taken to be $\sigma=0,10,30$, where $\sigma=0$ means no noise.
50\% and 80\% are experiments with 50\% and 80\% randomly missing pixels, respectively. Comparisons are among 3D tensor product complex tight framelet $\tpctf_6^{\reduced}$ with the redundancy rate $3\tfrac{5}{7}$, 3D tensor product complex tight framelet $\tpctf_6$ with the redundancy rate $29\tfrac{5}{7}$ (having the same directionality as $\tpctf_6^{\reduced}$), the 3D nonseparable
surfacelets in \cite{LD} with the redundancy rate $6.4$, the 3D nonseparable compactly supported shearlet frame $\mbox{DNST}^{3D}$-$42$ with the redundancy rates $42$. the 3D nonseparable compactly supported shearlet frame $\mbox{DNST}^{3D}$-$154$ with the redundancy rates $154$. The masks for inpainting  are $50\%$ or $80\%$ uniformly randomly missing pixels.
}\label{tab:video:inpainting}%
\end{center}
\end{small}
\end{table}%

From Table~\ref{tab:video:inpainting}, we see that the loss of performance of $\tpctf_6^{\reduced}$ is acceptable in comparison with $\tpctf_6$ for both \emph{Mobile} and \emph{Coastguard} in view of the redundancy rate of $\tpctf_6^{\reduced}$. Surfacelets do not perform well in the inpainting tests even though its redundancy rate is about twice of that of $\tpctf_6^\reduced$. When the missing pixels are 50\%, $\tpctf_6^\reduced$ outperforms $\mbox{DNST}^{3D}$-$42$ and $\mbox{DNST}^{3D}$-$154$, especially when the noise level is high ($\sigma=30$). When the missing pixels are 80\%, $\mbox{DNST}^{3D}$-$42$ and $\mbox{DNST}^{3D}$-$154$  have better performance with low noise level ($\sigma=0,10$) comparing to $\tpctf_6^\reduced$. However, when the noise level is high ($\sigma=30$), they no longer produce good results as $\tpctf_6^\reduced$ probably due to the reason that $\mbox{DNST}^{3D}$ employs undecimated transforms.

In summary, the proposed directional tensor product complex tight framelet $\tpctf_6^{\reduced}$ with low redundancy often performs better than other directional representation systems when an image or video is texture-rich, while it performs comparably with other directional representation systems for most other types of images and videos with significantly low redundancy rate of $\tpctf_6^{\reduced}$ in comparison with many other separable or nonseparable systems.


\begin{thebibliography}{10}



\bibitem{FDCT}
E. J. Cand\`{e}s, L. Demannet, D. Donoho, and L. Ying, Fast discrete curvelet transforms, \emph{Multiscale Model. Simul.} {\bf 5} (3) (2006) 861--899.


\bibitem{CCS:framelet:inpainting}
J.-F. Cai, R.~H. Chan, and Z.~Shen.
A framelet-based image inpainting algorithm.
{\em Appl. Comput. Harmon. Anal.},\textbf{24} (2008), 131--149.


\bibitem{CCS:ipi:2010}
J.-F. Cai, R.~H. Chan, and Z.~Shen.
Simultaneous cartoon and texture inpainting.
\emph{Inverse Problems Imaging}, \textbf{4} (2010), 379--395.

\bibitem{CDOS}
J.-F. Cai, B.~Dong, S.~Osher, and Z.~Shen, Image restoration: total variation, wavelet frames, and beyond. \emph{J. Amer. Math. Soc.} \textbf{25} (2012), 1033--1089.


\bibitem{CRSS}
R.~H.~Chan, S.~D.~Riemenschneider, L.~Shen, and Z.~Shen, Tight frame: an efficient way for high-resolution image reconstruction. \emph{Appl. Comput. Harmon. Anal.} \textbf{17} (2004), 91--115.

\bibitem{CHS} C. K. Chui,  W. He and  J. St\"ockler,
Compactly supported tight and sibling frames with maximum vanishing
moments, {\em Appl. Comput. Harmon. Anal.} {\bf 13} (2002), 224--262.

\bibitem{CoifDo}
R. R. Coifman and D. L. Donoho, Translation-invariant de-noising, in Wavelets and Statistics, Lecture Notes in Statistics, \textbf{103} (1995), 125--150, Springer-Verlag New York.



\bibitem{DHRS}
I.~Daubechies, B. ~Han, A.~Ron, and Z.~Shen, Framelets: MRA-based
constructions of wavelet frames, {\em Appl. Comput. Harmon. Anal.}
{\bf 14} (2003), 1--46.

\bibitem{DV}
M. N. Do and M. Vetterli, Contourlets, in G. V. Welland, editor,
Beyond Wavelets, Academic Press, (2008), 83--105.

\bibitem{DongShen:tutorial}
B.~Dong and Z.~Shen, MRA-based wavelet frames and applications, IAS/Park City Mathematics Series, \textbf{19}, (2010).


\bibitem{Han:acha:1997}
B.~Han, On dual wavelet tight frames, {\em Appl. Comput. Harmon.
Anal.}, \textbf{4} (1997), 380--413.

\bibitem{Han:acha:2010}
B.~Han, Pairs of frequency-based nonhomogeneous dual wavelet frames in the distribution space. \emph{Appl. Comput. Harmon. Anal.} \textbf{29} (2010), 330--353.

\bibitem{Han:acha:2012}
B.~Han, Nonhomogeneous wavelet systems in high dimensions, \emph{Appl. Comput. Harmon. Anal.} \textbf{32} (2012), 169--196.

\bibitem{Han:mmnp:2013}
B.~Han, Properties of discrete framelet transforms, \emph{Math. Model. Nat. Phenom.}
\textbf{8} (2013), 18--47.

\bibitem{Han:acha:2013}
B.~Han, Matrix splitting with symmetry and symmetric tight framelet filter banks with two high-pass filters, \emph{Appl. Comput. Harmon. Anal.}, \textbf{35} (2013), 200--227.

\bibitem{Han:acha:2014}
B. Han, Symmetric tight framelet filter banks with three high-pass filters,
\emph{Appl. Comput. Harmon. Anal.} \textbf{37} (2014), 140--161.

\bibitem{HanMoZhao:2013}
B.~Han, Q.~Mo, and Z.~Zhao, Compactly supported tensor product complex tight framelets with directionality, (2013), arXiv:1307.2599

\bibitem{HanZhao:siims:2014}
B.~Han and Z.~Zhao,
Tensor product complex tight framelets with increasing directionality, \emph{SIAM J. Imaging Sci.}, \textbf{7} (2014), 997--1034.

\bibitem{HanZhuang:acha:2014}
B.~Han and X.~Zhuang, Smooth affine shear tight frames with MRA structure, \emph{Appl. Comput. Harmon. Anal.} (2014), http://dx.doi.org/10.1016/j.acha.2014.09.005

\bibitem{K:acha:2001}
N. G. Kingsbury, Complex wavelets for shift invariant analysis and filtering of signals, \emph{Appl. Comput. Harmon. Anal.}, \textbf{10} (2001), 234--253.

\bibitem{KL}
G.~Kutyniok and D.~Labate, Shearlets: Multiscale Analysis for Multivariate Data, Birkh\"auser, 2012.

\bibitem{KSZ}
G. Kutyniok, M. Shahram, and X.~Zhuang,
ShearLab: a rational design of a digital parabolic scaling algorithm, \emph{SIAM J. Imaging Sci.} \textbf{5} (2012), 1291--1332.

\bibitem{LiShenSuter:ieee:2013}
Y.-R.~Li, L.~Shen, Lixin, and B.~W.~Suter, Adaptive inpainting algorithm based on DCT induced wavelet regularization. \emph{IEEE Trans. Image Process.} \textbf{22} (2013), 752--763.


\bibitem{LimIEEE}
W.-Q. Lim, The discrete shearlet transform: A new directional transform and compactly supported shearlet frames,
\emph{IEEE Trans. Image Process.} {\bf 19} (2010), 1166--1180.

\bibitem{LimIEEE2}
W.-Q. Lim, Nonseparable shearlet transform. \emph{IEEE Trans. Image Proc.} {\bf22} (2013), 2056 - 2065.


\bibitem{LD}
Y. M. Lu and M. N. Do, Multidimensional directional filter banks and
surfacelets, \emph{IEEE Trans. Image Process.}, \textbf{16} (2007), 918--931.

\bibitem{RonShen:twf}
A.~Ron and Z.~Shen, Affine systems in $L_2(\R^d)$: the analysis of the analysis operator, \emph{J. Funct. Anal.},  \textbf{148}  (1997),  408--447.



\bibitem{SBK}
I. W. Selesnick, R. G. Baraniuk, and N. G. Kingsbury, The dual-tree
complex wavelet transform, \emph{IEEE Signal Process. Mag.} {\bf 22}
(6) (2005) 123--151.


\bibitem{SS:bslocal}
L.~Sendur and I.~W.~Selesnick,
Bivariate shrinkage with local variance estimation,
\emph{IEEE Signal Process. Letters}, \textbf{9} (2002), 438--441.


\bibitem{SHB:2014}
Y. Shen, B. Han, and E. Braverman,
Image inpainting using directional tesnor product complex tight framelets, (2014), arXiv:1407.3234.
\end{thebibliography}
\end{document}